\documentclass[journal,onecolumn]{IEEEtran}
\usepackage{amsmath,amsthm}
\usepackage{booktabs}
\usepackage{latexsym}
\usepackage{mathrsfs}
\usepackage{amsthm}
\usepackage{amssymb}
\usepackage{amsfonts}
\usepackage{amsbsy}
\usepackage{tikz}
\usepackage{cases}

\usepackage[shortlabels]{enumitem}
\usepackage{url}
\usepackage{array}
\usepackage{pdflscape}
\usepackage{xcolor}
\usepackage{stmaryrd}
\usepackage{verbatim}
\usepackage{braket}
\usepackage{bbm}
\usepackage{bm}
\usepackage{rotating}
\usepackage{makecell}
\usepackage{cellspace}
\usepackage{threeparttable}
\usepackage [latin1]{inputenc}

\theoremstyle{plain}
\theoremstyle{definition}
\newtheorem{definition}{Definition}
\newtheorem{example}{Example}
\newtheorem{remark}{Remark}

\newtheorem{theorem}{Theorem}
\newtheorem{lemma}{Lemma}
\newtheorem{proposition}{Proposition}
\newtheorem{corollary}{Corollary}

\usepackage[colorlinks=true, citecolor=blue]{hyperref}

\allowdisplaybreaks[4]

\begin{document}
\title{Quantum codes and optimal pure quantum $(r,\delta)$-LRCs via the MP construction}
\author{\centerline{Meng Cao and Kun Zhou}
\thanks{M. Cao and K. Zhou are with Beijing Institute of Mathematical Sciences and Applications, Beijing, 101408, China. e-mails: mengcaomath@126.com (M. Cao),  kzhou@bimsa.cn (K. Zhou).}
\thanks{M. Cao is supported by the National Natural Science Foundation of China under Grant No. 12401684. K. Zhou is supported by the National Natural Science Foundation of China under Grant No. 12401040.}
}

\maketitle

\vspace{-10pt}

\begin{abstract}
The matrix-product (MP) construction method combines several classical codes of the same length with a defining matrix to form a new classical code known as an MP code. This approach provides an effective approach for constructing longer classical codes and various kinds of quantum codes.
Quantum $(r,\delta)$-locally recoverable codes (LRCs) are a new kind of quantum codes introduced in \cite{Galindo2026}, with significant potential for quantum data storage, where the optimal construction of pure quantum $(r,\delta)$-LRCs is a fundamental and compelling topic.
In this paper, we employ MP codes whose defining matrices are $\tau$-optimal defining ($\tau$-OD) matrices to construct new quantum codes and quantum $(r,\delta)$-LRCs. Specifically, we report the following results:
\vspace{2pt}
\begin{itemize}
\item We establish a unified $\tau$-monomial decomposition theorem for invertible self-adjoint matrices over finite fields of arbitrary characteristic, which
generalizes the result in \cite{CaoZhou2026} where the characteristic was required to be odd.
Based on this theorem, we prove the existence of $\tau$-OD matrices over $\mathbb{F}_{q^2}$ for any characteristic and demonstrate that there exist several new infinite families of $\tau$-OD matrices over $\mathbb{F}_{q^2}$ of characteristic $2$.
As an application of MP codes involving $\tau$-OD matrices, we construct several infinite families of quantum codes with flexible parameters. Within this framework, we present $222$ record-breaking quantum codes that surpass the best-known records maintained in Grassl's database \cite{Grassl2025Bounds}.

\vspace{3pt}

\item We propose two effective schemes for constructing optimal pure quantum $(r,\delta)$-LRCs via MP codes.
Accordingly, we construct four new infinite families of optimal pure quantum $(r,\delta)$-LRCs with flexible parameters.
Notably, we report an interesting phenomenon by exhibiting $30$ optimal pure quantum $(r,\delta)$-LRCs derived from our framework; that is, there exist quantum codes that are not only optimal pure quantum $(r,\delta)$-LRCs but also, according to Grassl's database \cite{Grassl2025Bounds}, best-known, optimal, or record-breaking quantum codes.
To the best of our knowledge, the new discovery that quantum codes are simultaneously optimal pure quantum $(r,\delta)$-LRCs and record-breaking quantum codes has not been previously reported in the literature.
\end{itemize}
\end{abstract}


\section{Introduction}\label{sec:intro}

Unlike classical bits, which can be easily replicated and are relatively robust, quantum states are inherently fragile and highly susceptible to environmental disturbances. Furthermore, the no-cloning theorem prohibits the perfect copying of arbitrary quantum states, rendering naive classical redundancy techniques inapplicable in the quantum setting. To overcome these limitations, quantum information is encoded into a larger Hilbert space, enabling error detection and correction without collapsing the logical state. Consequently, quantum error-correcting codes (QECCs) have become indispensable for mitigating decoherence and noise in quantum information processing.

The theory of QECCs originated with Shor's pioneering work \cite{Shor1995Scheme}, which introduced the $[[9,1,3]]$ code, the first scheme capable of correcting arbitrary single-qubit errors. In \cite{Steane1996Error}, Steane introduced the $[[7,1,3]]$ code, representing another major milestone.
In \cite{Gottesman1997Stabilizer}, Gottesman developed the stabilizer formalism, which established a systematic framework for quantum code construction and laid the foundation for much of modern quantum error correction.
Currently, the majority of studied quantum codes fall within the framework of stabilizer codes, primarily due to their profound links to classical codes.
Prominent examples include the CSS construction \cite{Calderbank1996Good,Steane1996Simple} and the Hermitian construction \cite{Ashikhmin2001Nonbinary}, both of which enable the efficient derivation of QECCs from suitable classical codes.

Recently, Galindo, Hernando, Mart\'{i}n-Cruz, and Matsumoto in \cite{Galindo2026} introduced \emph{quantum $(r,\delta)$-locally recoverable codes (LRCs)}.
These codes represent a novel class of quantum codes with significant potential for quantum data storage.
A quantum $(r,\delta)$-LRC can correct $\delta-1$ qudit erasures from sets of at most $r + \delta - 1$ qudits.
In the same work, they established a fundamental connection between Hermitian (resp. Euclidean) dual-containing codes and quantum $(r,\delta)$-LRCs.
Their work effectively explored general quantum $(r,\delta)$-LRCs and optimal pure quantum $(r,\delta)$-LRCs induced by classical codes.
In \cite{Galindo2025Optimal} and \cite{Galindo2026QuantumBCH}, the authors further investigated quantum $(r,\delta)$-LRCs and provided new optimal constructions of pure quantum $(r,\delta)$-LRCs.

Let $\mathbb{F}_{q}^{k\times t}$ denote the set of all $k\times t$ matrices over the finite field $\mathbb{F}_{q}$ of order $q$, where $q$ is a prime power.
In \cite{Blackmore2001Matrix}, Blackmore and Norton introduced \emph{matrix-product (MP) codes}, a class of classical codes formed by combining a defining matrix with several shorter classical codes of the same length. Another early contribution to MP codes was made by {\"O}zbudak and Stichtenoth in \cite{Ozbudak2002Note}.
Let $\mathcal{C}_{1},\mathcal{C}_{2},\ldots,\mathcal{C}_{k}$ be linear codes of length $n$ over $\mathbb{F}_{q}$.
Let $A=(a_{i,j})\in\mathbb{F}_{q}^{k\times t}$ with $k\leq t$. The MP code associated with $\mathcal{C}_i$ and $A$ is denoted by
\begin{align*}
\mathcal{C}(A):=[\mathcal{C}_{1},\mathcal{C}_{2},\ldots,\mathcal{C}_{k}]\cdot A,
\end{align*}
which consists of all matrix-products of the form $[\mathbf{c}_{1},\mathbf{c}_{2},\ldots,\mathbf{c}_{k}]\cdot A$, where $\mathbf{c}_{i}\in\mathcal{C}_{i}$ for $1\leq i\leq k$.
Here, $A$ is referred to as the \emph{defining matrix} of $\mathcal{C}(A)$, and $\mathcal{C}_{1},\mathcal{C}_{2},\ldots,\mathcal{C}_{k}$ are called the \emph{constituent codes} of $\mathcal{C}(A)$. Following \cite{Blackmore2001Matrix} and \cite{Hernando2009Construction}, every codeword $\mathbf{c}=[\mathbf{c}_{1},\mathbf{c}_{2},\ldots,\mathbf{c}_{k}]\cdot A\in\mathcal{C}(A)$ can be expressed as a row vector of length $tn$ of the form $\mathbf{c}=\big(\sum_{i=1}^k a_{i,1}\mathbf{c}_{i},\sum_{i=1}^k a_{i,2}\mathbf{c}_{i},\ldots,\sum_{i=1}^k a_{i,t}\mathbf{c}_{i}\big)$, where $\mathbf{c}_{i}=(c_{1,i},c_{2,i},\ldots,c_{n,i})\in\mathcal{C}_{i}$ for $1\leq i\leq k$.

Let $A(j_{1},j_{2},\ldots,j_{i})$ denote the submatrix consisting of the first $i$ rows and the $j_{1}, j_{2}, \ldots, j_{i}$-th columns of $A\in \mathbb{F}_{q}^{k\times t}$, where $k\leq t$, $1\leq i\leq k$ and $1\leq j_{1}<j_{2}<\ldots< j_{i}\leq t$. If $A(j_{1},j_{2},\ldots,j_{i})$ is invertible for all $1\leq i\leq k$ and $1\leq j_{1}<j_{2}<\ldots<j_{i}\leq t$, then $A$ is called a {\it non-singular by columns (NSC) matrix}.
As shown in \cite{Blackmore2001Matrix} and reviewed in Lemma \ref{proposition2}, the MP code guarantees an explicit lower bound on its minimum distance if its defining matrix is an NSC matrix. This prompts a deeper investigation into NSC matrices.

The following two propositions, provided in \cite[Proposition 3.3]{Blackmore2001Matrix} and \cite[Lemma 3.1]{Cao2025Entanglement} respectively, establish the existence condition and the multiplicative invariance of NSC matrices.

\begin{proposition}\label{existNSC}
For $s\geq 2$, there exists an $s\times k$ NSC matrix over $\mathbb{F}_{q}$ if and only if $s\leq k\leq q$.
\end{proposition}

\begin{proposition}\label{LND}
Let $N\in\mathbb{F}_{q}^{k\times k}$ be an NSC matrix. Then, for any invertible lower triangular matrix $L\in\mathbb{F}_{q}^{k\times k}$ and invertible diagonal matrix $D\in\mathbb{F}_{q}^{k\times k}$, the matrices $LN$, $ND$, and $LND$ are all NSC matrices.
\end{proposition}

The matrix-product (MP) construction technique serves as a powerful tool for constructing various kinds of quantum codes.
MP codes whose defining matrices are NSC matrices, especially those possessing special structural properties such as orthogonality, unitarity, quasi-orthogonality, or quasi-unitarity, are of significant value in the construction of different types of quantum codes. For instance, Galindo, Hernando, and Ruano in \cite{Galindo2015New} constructed Euclidean dual-containing MP codes via NSC orthogonal, NSC quasi-orthogonal, and quasi-orthogonal matrices, which yielded numerous new quantum codes. In \cite{Liu2018On}, Liu, Dinh, Liu, and Yu established an efficient method for constructing Hermitian dual-containing MP codes based on quasi-unitary matrices, resulting in many new quantum codes through the application of NSC quasi-unitary and quasi-unitary matrices. MP codes whose defining matrices are NSC quasi-unitary matrices, NSC quasi-orthogonal matrices, or their related derivatives, have proven very useful in constructing entanglement-assisted quantum codes (see, e.g., \cite{Cao2025Classical}, \cite{Cao2024Construction}, \cite{Cao2025Entanglement}, \cite{Li2025EAQEC}, \cite{Liu2019Entanglement}). Furthermore, Liu and Hu in \cite{Liu2024Constructions} provided new asymmetric entanglement-assisted quantum codes via MP codes associated with NSC quasi-orthogonal and quasi-orthogonal matrices. Very recently, Liu, Hu, and Liu in \cite{Liu2026subsystem} utilized MP codes involving NSC quasi-orthogonal and quasi-orthogonal matrices to construct new subsystem codes.

The concept of \emph{$\tau$-optimal defining ($\tau$-OD) matrices} was introduced in \cite{Cao2024On} as a special class of NSC matrices.
Such matrices were subsequently categorized into two types in \cite{Cao2025Entanglement}, namely type I $\tau$-OD matrices and type II $\tau$-OD matrices.
In this paper, we focus on type II $\tau$-OD matrices to construct our desired codes.
For the sake of brevity, the term ``$\tau$-OD matrix'' will exclusively denote the ``type II $\tau$-OD matrix' throughout this paper.
Specifically, a matrix $A=(a_{i,j})\in \mathbb{F}_{q^{2}}^{k \times k}$ is called a {\it $\tau$-OD matrix} if $A$ is an NSC matrix and $AA^{\dag}$ is a $\tau$-monomial matrix; that is, $AA^{\dag}=DP_{\tau}$, where $D$ is an invertible diagonal matrix and $P_{\tau}$ is the permutation matrix associated with the permutation $\tau$. Here, $A^{\dag}=(a_{j,i}^{q})$ denotes the conjugate transpose of $A=(a_{i,j})$. When $\tau$ is the identity permutation, the corresponding $\tau$-OD matrix reduces to an NSC quasi-unitary matrix. Consequently, $\tau$-OD matrices generalize NSC quasi-unitary matrices.

In this paper, we employ MP codes whose defining matrices are $\tau$-OD matrices to construct new quantum codes and quantum $(r,\delta)$-LRCs.
Our main contributions are summarized as follows:
\begin{itemize}
\item [1.] By applying the matrix properties developed in this paper for finite fields of arbitrary characteristic to $\tau$-OD matrices and MP codes,
we obtain numerous record-breaking quantum codes. Specifically,
\begin{itemize}
\item [1)] We establish a unified $\tau$-monomial decomposition theorem for invertible self-adjoint matrices over finite fields of arbitrary characteristic,
generalizing the result in \cite{CaoZhou2026} where the characteristic was required to be odd (see Theorem \ref{theorem-com}).

\vspace{3pt}

\item [2)] We prove the existence of $\tau$-OD matrices over $\mathbb{F}_{q^2}$ for any characteristic (see Theorem \ref{1theorem-com}).
We show that there exist several new infinite families of $\tau$-OD matrices over finite fields of characteristic $2$ (see Theorems \ref{theorem-inf1} and \ref{theorem-inf2}).

\vspace{3pt}

\item [3)] Using MP codes involving $\tau$-OD matrices, we derive several infinite families of quantum codes with flexible parameters.
Within this framework, we present $222$ record-breaking quantum codes that surpass the best-known lower bounds on the minimum distances of quantum codes maintained in Grassl's database \cite{Grassl2025Bounds} (see Tables \ref{table1} and \ref{table2}).
\end{itemize}

\vspace{3pt}

\item [2.] MP codes with $\tau$-OD matrices as their defining matrices are utilized to the optimal construction of pure quantum $(r,\delta)$-LRCs. Specifically,
\begin{itemize}
\item [1)] We propose two effective schemes for constructing optimal pure quantum $(r,\delta)$-LRCs via MP codes (see Theorems \ref{theorem-optqx} and
\ref{theorem-optqx-x1}).

\vspace{3pt}

\item [2)] Based on Theorems \ref{theorem-optqx} and \ref{theorem-optqx-x1}, we construct four new infinite families of optimal pure quantum
$(r,\delta)$-LRCs with flexible parameters (see Theorems \ref{cor-sepoql}-\ref{3qsquare}).
The parameters of our optimal pure quantum $(r,\delta)$-LRCs established in Theorems \ref{cor-sepoql}-\ref{3qsquare} are new in comparison with the existing constructions in \cite{Galindo2026, Galindo2025Optimal, Galindo2026QuantumBCH} (see the detailed comparison in Subsection \ref{compar} and Table \ref{table666}).

\vspace{3pt}

\item [3)] We report an interesting phenomenon by exhibiting $30$ optimal pure quantum $(r,\delta)$-LRCs derived from our framework; that is, there exist quantum
codes that are not only optimal pure quantum $(r,\delta)$-LRCs but also, according to Grassl's database \cite{Grassl2025Bounds}, best-known, optimal, or record-breaking quantum codes (see Table \ref{table3} and Remark \ref{remark1rc}). To the best of our knowledge, the new discovery that quantum codes are simultaneously optimal pure quantum $(r,\delta)$-LRCs and record-breaking quantum codes has not been previously reported in the literature.
\end{itemize}
\end{itemize}

\vspace{6pt}

This paper is organized as follows. Section \ref{sec:pre} provides the necessary preliminaries for this paper.
Section \ref{sec:com} establishes a unified $\tau$-monomial decomposition theorem for invertible self-adjoint matrices and proves the existence of new $\tau$-OD matrices. Section \ref{infinitequ} presents several infinite families of quantum codes, including record-breaking quantum codes.
Section \ref{qopt} constructs several infinite families of optimal pure quantum $(r,\delta)$-LRCs and reports a new discovery regarding such codes.
Finally, Section \ref{conclusion} concludes the paper.

\section{Preliminaries}\label{sec:pre}

An $[n,k,d]_{q}$ linear code over $\mathbb{F}_{q}$ is a $k$-dimensional subspace of $\mathbb{F}_{q}^n$ with minimum (Hamming) distance $d$.
By the {\it Singleton bound}, the minimum distance satisfies $d\leq n-k+1$.
A code achieving equality, i.e., $d=n-k+1$, is called a {\it maximum distance separable (MDS) code}.

The {\it Euclidean dual code} $\mathcal{C}^{\perp_{\mathrm{E}}}$ and the {\it Hermitian dual code} $\mathcal{C}^{\perp_{\mathrm{H}}}$ of a linear code $\mathcal{C}$ of length $n$ over $\mathbb{F}_{q^{2}}$ are respectively defined as
\begin{align*}
\mathcal{C}^{\perp_{\mathrm{E}}}=\Big\{\mathbf{x}=(x_{1},\ldots,x_{n})\in\mathbb{F}_{q^{2}}^{n}:\
\langle\mathbf{x},\mathbf{y}\rangle_{\mathrm{E}}=\sum_{i=1}^n x_i y_i=0 \ \mathrm{for} \ \mathrm{all} \  \mathbf{y}=
(y_{1},\ldots,y_{n})\in \mathcal{C}\Big\}
\end{align*}
and
\begin{align*}
\mathcal{C}^{\perp_{\mathrm{H}}}=\Big\{\mathbf{x}=(x_{1},\ldots,x_{n})\in\mathbb{F}_{q^{2}}^{n}:\
\langle\mathbf{x},\mathbf{y}\rangle_{\mathrm{H}}=\sum_{i=1}^n x_i y_i^{q}=0 \ \mathrm{for} \ \mathrm{all} \  \mathbf{y}=
(y_{1},\ldots,y_{n})\in \mathcal{C}\Big\}.
\end{align*}

We say that $\mathcal{C}$ is {\it Euclidean self-orthogonal} if $\mathcal{C}\subseteq\mathcal{C}^{\perp_{\mathrm{E}}}$,
and {\it Hermitian self-orthogonal} if $\mathcal{C}\subseteq\mathcal{C}^{\perp_{\mathrm{H}}}$.
Conversely, $\mathcal{C}$ is called {\it Euclidean dual-containing} if $\mathcal{C}^{\perp_{\mathrm{E}}}\subseteq\mathcal{C}$,
and {\it Hermitian dual-containing} if $\mathcal{C}^{\perp_{\mathrm{H}}}\subseteq\mathcal{C}$.

\subsection{Generalized Reed-Solomon (GRS) Codes}

Let $k$ and $n$ be two positive integers such that $k\leq n\leq q$.
Given two vectors $\mathbf{a}=(a_{1},\ldots,a_{n})$ and $\mathbf{v}=(v_{1},\ldots,v_{n})$, where $a_{1},\ldots,a_{n}$ are distinct elements in $\mathbb{F}_{q}$
and $v_{1},\ldots,v_{n}$ are nonzero elements in $\mathbb{F}_{q}$.
The \emph{generalized Reed-Solomon (GRS) code} of length $n$ and dimension $k$, associated with $\mathbf{a}$ and $\mathbf{v}$, is defined as
\begin{align*}
\mathrm{GRS}_{k}(\mathbf{a},\mathbf{v})=\{(v_{1}f(a_{1}),\ldots,v_{n}f(a_{n})): f(x)\in \mathbb{F}_{q}[x], \mathrm{with}\ \mathrm{deg}(f(x))\leq k-1\}.
\end{align*}
This code is an $[n, k]_q$ MDS code with a generator matrix given by
\begin{align}\label{GRSGenerator}
G_{k}(\mathbf{a},\mathbf{v})=\begin{pmatrix}
v_{1}& v_{2}&  \cdots &v_{n}\\
v_{1}a_{1}&v_{2}a_{2}& \cdots &v_{n}a_{n}\\
\vdots&\vdots&\ddots&\vdots\\
v_{1}a_{1}^{k-1}& v_{2}a_{2}^{k-1}& \ldots &v_{n}a_{n}^{k-1}
\end{pmatrix}.
\end{align}

\subsection{Matrix-Product (MP) Codes}

We recall several well-known results on MP codes in different settings.

\begin{lemma}{\rm (\cite[Page 54]{Ozbudak2002Note})}\label{proposition1}
Let $\mathcal{C}_{i}$ be an $[n,r_{i},d_{i}]_{q}$ linear code for $1\leq i\leq k$, and let $A\in \mathbb{F}_{q}^{k\times t}$ be a matrix of full row rank.
Then, the MP code $\mathcal{C}(A)=[\mathcal{C}_{1},\mathcal{C}_{2},\ldots,\mathcal{C}_{k}]\cdot A$ is an
$\big[tn,\sum_{i=1}^{k}r_{i},\geq \min\{D_{i}(A)d_{i}:1\leq i\leq k\}\big]_{q}$ linear code, where $D_{i}(A)$ denotes the minimum distance of the linear code on $\mathbb{F}_{q}^{t}$ generated by the first $i$ rows of $A$.
\end{lemma}

\begin{lemma}{\rm (\cite[Theorem 1]{Hernando2009Construction})}\label{nested-distance}
With the same notation as in Lemma \ref{proposition1}. If the constituent codes $\mathcal{C}_i$ are nested for $1\leq i\leq k$, i.e.,
$\mathcal{C}_1\supseteq\mathcal{C}_2\supseteq \ldots\supseteq \mathcal{C}_{k}$, then the minimum distance of $\mathcal{C}(A)$ is $d(\mathcal{C}(A))=\min\{D_{i}(A)d_{i}: 1\leq i\leq k\}$.
\end{lemma}

The Singleton bound implies that for any matrix $A\in \mathbb{F}_{q}^{k\times t}$ of full row rank, $D_{i}(A)\leq t-i+1$ holds for all $1\leq i\leq k$.
By the definition of NSC matrices, if $A\in \mathbb{F}_{q}^{k\times t}$ is NSC, then $D_{i}(A)=t-i+1$ for all $1\leq i\leq k$.
Consequently, the parameters of an MP code whose defining matrix is an NSC matrix are given in the following lemma.

\begin{lemma}{\rm (\cite[Theorem 3.7]{Blackmore2001Matrix})}\label{proposition2}
Let $\mathcal{C}_{i}$ be an $[n,r_{i},d_{i}]_{q}$ linear code for $1\leq i\leq k$, and let $A\in \mathbb{F}_{q}^{k\times t}$ be an NSC matrix.
Then, the MP code $\mathcal{C}(A)=[\mathcal{C}_{1},\mathcal{C}_{2},\ldots,\mathcal{C}_{k}]\cdot A$ is an
$\big[tn,\sum_{i=1}^{k}r_{i},\geq \min\{(t-i+1)d_{i}:1\leq i\leq k\}\big]_{q}$ linear code.
\end{lemma}

\begin{lemma}{\rm (\cite[Page 7]{Zhang2015Quantum})}\label{proposition4}
Let $\mathcal{C}_{i}$ be a linear code of length $n$ over $\mathbb{F}_{q^{2}}$ for $1\leq i\leq k$, and let $A\in\mathbb{F}_{q^{2}}^{k\times k}$ be an invertible matrix.
Then, the Hermitian dual code of the MP code $\mathcal{C}(A)=[\mathcal{C}_{1},\mathcal{C}_{2},\ldots,\mathcal{C}_{k}]\cdot A$ is given by
\begin{align*}
\mathcal{C}(A)^{\bot_{\mathrm{H}}}=\big[\mathcal{C}_{1}^{\bot_{\mathrm{H}}},\mathcal{C}_{2}^{\bot_{\mathrm{H}}},\ldots, \mathcal{C}_{k}^{\bot_{\mathrm{H}}}\big]
\cdot (A^{-1})^{\dag}.
\end{align*}
Here $B^{\dag}=(b_{j,i}^{q})$ denotes the conjugate transpose of the matrix $B=(b_{i,j})$ over $\mathbb{F}_{q^{2}}$.
\end{lemma}

\subsection{Classical $(r,\delta)$-Locally Recoverable Codes (LRCs)}\label{LRC}

Given an $[n,k,d]_q$ linear code $\mathcal{C}$. For $1\leq i\leq n$, the $i$-th symbol $c_i$ of $\mathcal{C}$ is said to have \emph{$(r,\delta)$-locality} if there exists a subset $S_i\subseteq [n]$ containing $i$ and a punctured code $\mathcal{C}|_{S_i}$ such that $|S_i|\leq r+\delta-1$ and the minimum distance $d(\mathcal{C}|_{S_i})\geq \delta $. Here, $\mathcal{C}|_{S_i}$ denotes the code $\mathcal{C}$ punctured on the coordinate set $[n]\backslash S_i$ by deleting the components indexed by $[n]\backslash S_i$ in each codeword of $\mathcal{C}$. The code $\mathcal{C}$ is called a \emph{classical $(r,\delta)$-LRC} if all its symbols have $(r,\delta)$-locality.

In \cite{Prakash2012}, a Singleton-type bound for the parameters of a classical $(r,\delta)$-LRC was established:
\begin{align}\label{rdelta-singleton}
d\leq n-k+1-\bigg(\bigg\lceil \frac{k}{r}\bigg\rceil-1\bigg)(\delta-1),
\end{align}
where $\lceil\cdot\rceil$ denotes the ceiling function.
A classical $(r,\delta)$-LRC attaining this bound with equality is called an \emph{optimal classical $(r,\delta)$-LRC}.

\begin{lemma}{\rm (\cite[Lemma 3]{Luo2023Three})}\label{local}
Let $\mathcal{C}_i$ be a linear code of length $n$ over $\mathbb{F}_{q}$ for $1\leq i\leq k$, with
$\mathcal{C}_1\supseteq\mathcal{C}_2\supseteq \ldots\supseteq \mathcal{C}_{k}$.
Let $A\in \mathbb{F}_{q}^{k\times t}$ be a matrix of full row rank. If $\mathcal{C}_1$ is a classical $(r,\delta)$-LRC, then the MP code $\mathcal{C}(A)=[\mathcal{C}_{1},\mathcal{C}_{2},\ldots,\mathcal{C}_{k}]\cdot A$ is also a classical $(r,\delta)$-LRC.
\end{lemma}

\subsection{Quantum Codes and Quantum $(r,\delta)$-Locally Recoverable Codes (LRCs)}

A $q$-ary quantum code $[[n,k,d]]_q$ is defined as a $q^k$-dimensional subspace of the complex Hilbert space $(\mathbb{C}^q)^{\otimes n} \cong \mathbb{C}^{q^n}$, where $d$ is the minimum distance. It is capable of detecting up to $d-1$ quantum errors and correcting up to $\left\lfloor \frac{d-1}{2} \right\rfloor$ quantum errors, where $\lfloor\cdot\rfloor$ denotes the floor function. The parameters of an $[[n,k,d]]_{q}$ quantum code satisfy the \emph{quantum Singleton bound}, i.e., $2d \leq n-k+2$.
Codes that attain this bound, i.e., $2d=n-k+2$, are referred to as \emph{quantum MDS codes}.

One of the most effective methods for constructing quantum codes is the {\it Hermitian construction} (see \cite[Corollary 1]{Ashikhmin2001Nonbinary}), which yields quantum codes from Hermitian dual-containing codes.

\begin{lemma}{\rm (\cite[Corollary 1]{Ashikhmin2001Nonbinary})}\label{proposition6}
Let $\mathcal{C}$ be an $[n,k,d]_{q^{2}}$ Hermitian dual-containing code. Then, there exists a quantum code $\mathcal{Q}$ with parameters
$[[n,2k-n,d(\mathcal{Q})\geq d]]_{q}$, where
\begin{align*}
d(\mathcal{Q})&=\begin{cases}
d, & \mathrm{if}\ \mathcal{C}^{\bot_{\mathrm{H}}}=\mathcal{C},\\
\mathrm{wt}(\mathcal{C}\backslash\mathcal{C}^{\bot_{\mathrm{H}}}), & \mathrm{if}\ \mathcal{C}^{\bot_{\mathrm{H}}}\subsetneq\mathcal{C}.
\end{cases}
\end{align*}
Here, $\mathrm{wt}(\mathcal{C}\backslash\mathcal{C}^{\bot_{\mathrm{H}}})$ represents the minimum weight of all nonzero vectors in
$\mathcal{C}\backslash\mathcal{C}^{\bot_{\mathrm{H}}}$.
The quantum code $\mathcal{Q}$ is called \textit{pure} if $d(\mathcal{Q})=d$, and \textit{impure} otherwise.
\end{lemma}

Starting from existing quantum codes, one can obtain new ones by employing the propagation rules, such as the lengthening and subcode constructions
(see \cite[Proposition 2.8]{Feng2006Asymptotic} and \cite[Theorem 4.1]{Grassl2021Algebraic}).

\begin{lemma}\label{proposition66}
Let $\mathcal{Q}$ be an $[[n,k,d]]_{q}$ quantum code. Then,
\begin{itemize}
\item [(1)] \textbf{\textup{(Lengthening)}} there exists an $[[n+1,k,\geq d]]_{q}$ quantum code if $k>0$.

\item [(2)] \textbf{\textup{(Subcode)}} there exists an $[[n,k-1,\geq d]]_{q}$ quantum code if $k>1$ or if $k=1$ and the initial code $\mathcal{Q}$ is pure.
\end{itemize}
\end{lemma}

In a recent work \cite{Galindo2026}, Galindo, Hernando, Mart\'{i}n-Cruz, and Matsumoto introduced \emph{quantum $(r,\delta)$-locally recoverable codes (LRCs)},
a novel class of quantum codes with significant potential for quantum data storage.
In the same work, they established a fundamental framework connecting Hermitian (resp. Euclidean) dual-containing codes to quantum $(r,\delta)$-LRCs.
Their results provide explicit characterizations for these codes, including the optimality of pure quantum $(r,\delta)$-LRCs constructed from classical codes.

\begin{lemma}{\rm (\cite[Theorem 31]{Galindo2026})}\label{quan-LRC}
Let $\mathcal{C}\subseteq \mathbb{F}_{q^2}^n$ (resp. $\mathcal{C}\subseteq \mathbb{F}_{q}^n$) be a Hermitian (resp. Euclidean) dual-containing code, $\dim(\mathcal{C})=\frac{n+k}{2}$ and $\mathcal{C}$ be a classical $(r,\delta)$-LRC. Then, it induces a quantum $(r,\delta)$-LRC $Q'(\mathcal{C})$ with parameters $[[n,k,\geq d(\mathcal{C})]]_{q}$, which satisfies
\begin{align}\label{equation1211}
k+2d(\mathcal{C})+2\bigg(\bigg\lceil\frac{n+k}{2r}\bigg\rceil-1\bigg)(\delta-1)\leq n+2.
\end{align}
\end{lemma}

\begin{definition}{\rm (\cite[Definition 33]{Galindo2026})}\label{optipure}
Following the notation of Lemma \ref{quan-LRC}, a pure quantum $(r,\delta)$-LRC $Q'(\mathcal{C})$ with parameters $[[n,k,d(\mathcal{C})]]_{q}$ is called an
\emph{optimal pure quantum $(r,\delta)$-LRC} if
\begin{align}\label{quan-inqu}
k+2d(\mathcal{C})+2\bigg(\bigg\lceil\frac{n+k}{2r}\bigg\rceil-1\bigg)(\delta-1)=n+2.
\end{align}
\end{definition}

\begin{remark}\label{optimal}
Due to the condition $k=2\dim(\mathcal{C})-n$ in Lemma \ref{quan-LRC}, one can easily verify that Eq. \eqref{quan-inqu} is equivalent to
\begin{align*}
d(\mathcal{C})=n-\dim(\mathcal{C})+1-\bigg(\bigg\lceil \frac{\dim(\mathcal{C})}{r}\bigg\rceil-1\bigg)(\delta-1).
\end{align*}
This means that $\mathcal{C}$ is an optimal classical $(r,\delta)$-LRC. Consequently, within the framework of Lemma \ref{quan-LRC} and Definition \ref{optipure},
a Hermitian (resp. Euclidean) dual-containing classical $(r,\delta)$-LRC is optimal if and only if its induced pure quantum $(r,\delta)$-LRC is optimal.
\end{remark}

\section{Unified $\tau$-Monomial Decomposition Theorem and New $\tau$-OD Matrices} \label{sec:com}

In this section, we establish a unified $\tau$-monomial decomposition theorem for invertible self-adjoint matrices over finite fields of arbitrary characteristic.
Based on our decomposition theorem, we prove the existence of $\tau$-OD matrices over $\mathbb{F}_{q^2}$ for any characteristic and demonstrate that there exist several new infinite families of $\tau$-OD matrices over $\mathbb{F}_{q^2}$ of characteristic $2$.

\subsection{Notation and Concepts}

We need to use the following notation and concepts:
\begin{itemize}
\item Let $\sigma$ be an automorphism of $\mathbb{F}_{q}$ such that $\sigma^2={\rm id}_{\mathbb{F}_{q}}$, where ${\rm id}_{\mathbb{F}_{q}}$ denotes the identity
automorphism of $\mathbb{F}_{q}$.

\vspace{3pt}

\item For any matrix $A=(a_{i,j})$ over $\mathbb{F}_{q}$, we define $A^{\sigma}=(\sigma(a_{i,j}))$ and $A^{\star}=(A^{\top})^{\sigma}$, where $A^{\top}$ denotes the
\emph{transpose} of $A$. A matrix $A$ is said to be \emph{self-adjoint} if $A^{\star}=A$. For any row vector $\mathbf{a}=(a_{1},\ldots,a_{k})\in\mathbb{F}_{q}^{k}$, we define $\mathbf{a}^{\sigma}=(\sigma(a_{1}),\ldots,\sigma(a_{k}))$ and $\mathbf{a}^{\star}=(\sigma(a_{1}),\ldots,\sigma(a_{k}))^{\top}$.

\vspace{3pt}

\item Let ${\rm Sym}_{k}$ denote the \emph{symmetric group} of degree $k$, which consists of all permutations of the set $\{1, 2, \ldots, k\}$.
A permutation $\tau\in{\rm Sym}_{k}$ is called the \emph{identity permutation} if $\tau(i)=i$ for all $1\leq i\leq k$, and we denote it by $\mathbbm{1}_{k}$.
For any $1\leq a<b\leq k$, we denote by $(a,b)$ a \emph{2-cycle} in ${\rm Sym}_{k}$.
Specifically, if $\tau=(a,b)\in{\rm Sym}_{k}$, then $\tau(a)=b$, $\tau(b)=a$ and $\tau(i)=i$ for all $i\in\{1,2,\ldots,k\}\backslash\{a,b\}$.

\vspace{3pt}

\item Let $I_{s}$ denote the $s\times s$ identity matrix. Let $O_{s\times t}$ (resp. $O_{t\times s}$) denote the $s\times t$ (resp. $t\times s$) zero matrix.

\vspace{3pt}

\item Let $\mathbf{0}_{1\times s}$ (resp. $\mathbf{0}_{s\times 1}$) denote the $1\times s$ (resp. $s\times 1$) zero vector.

\vspace{3pt}

\item A matrix $A\in \mathbb{F}_{q}^{k\times k}$ is called a \emph{$\tau$-monomial matrix} if we can write $A=DP_{\tau}$, where $D\in \mathbb{F}_{q}^{k\times k}$ is
an invertible diagonal matrix and $P_{\tau}$ is a permutation matrix with respect to the permutation $\tau\in\mathrm{Sym}_{k}$ such that the $\tau(i)$-th row of $P_{\tau}$ corresponds to the $i$-th row of the identity matrix $I_{k}$ for all $1\leq i\leq k$.

\vspace{3pt}

\item Let $A=(a_{i,j})_{1\leq i,j\leq k}\in\mathbb{F}_{q}^{k\times k}$.
If $a_{i,j}=0$ for all $1\leq i<j\leq k$ (resp. $1\leq j<i\leq k$), then $A$ is called a \emph{lower triangular} (resp. an \emph{upper triangular}) \emph{matrix}.
If a lower triangular (resp. an upper triangular) matrix $A$ satisfies $a_{i,i}=1$ for all $1\leq i\leq k$, then it is called a \emph{unit lower triangular}
(resp. a \emph{unit upper triangular}) \emph{matrix}.
\end{itemize}

\subsection{Unified $\tau$-Monomial Decomposition Theorem}

We first consider the finite field extension $\mathbb{E}=\mathbb{F}_{q^a}\supseteq \mathbb{F}_q=\mathbb{F}$, where $q$ is a prime power and $a$ is a positive integer.
The \emph{trace function} $\mathrm{Tr}_{\mathbb{E}/\mathbb{F}}:\mathbb{E}\rightarrow \mathbb{F}$ is defined as
\begin{align*}
\mathrm{Tr}_{\mathbb{E}/\mathbb{F}}(x)=\sum_{i=0}^{a-1} x^{q^i},\;x\in \mathbb{E}.
\end{align*}

\begin{proposition}{\rm (\cite[Theorem 2.25]{Lidl1994Introduction})}\label{theorem-trace}
Suppose $\mathbb{E}$ is a finite extension of $\mathbb{F}=\mathbb{F}_q$ and $\alpha\in \mathbb{E}$. Then, $\mathrm{Tr}_{\mathbb{E}/\mathbb{F}}(\alpha)=0$ if and only if there exists $\beta\in \mathbb{E}$ such that $\alpha=\beta^q-\beta$.
\end{proposition}

If we consider finite fields of characteristic $2$, the following result can be readily obtained.

\begin{lemma}\label{cor-trace}
Let $q=2^a$, where $a$ is a positive integer. Then, for any $\alpha \in \mathbb{F}_{q}$, there exists $\beta\in \mathbb{F}_{q^2}$ such that $\alpha=\beta^q+\beta$.
\end{lemma}

\begin{IEEEproof}
Since $\mathbb{F}_{q}$ has characteristic $2$, it follows that
\begin{align*}
\mathrm{Tr}_{\mathbb{F}_{q^{2}}/\mathbb{F}_{q}}(\alpha)=0.
\end{align*}
Applying Proposition \ref{theorem-trace} and again using the fact that $\mathbb{F}_{q}$ has characteristic $2$, we conclude that there exists $\beta\in \mathbb{F}_{q^{2}}$ such that $\alpha=\beta^q+\beta$. This completes the proof.
\end{IEEEproof}

\vspace{6pt}

Now, let us introduce the concept of the $\tau$-monomial decomposition for invertible self-adjoint matrices.

\begin{definition}\textbf{\textup{($\tau$-Monomial Decomposition)}}\label{def-UL}
Let $\sigma$ be an automorphism of $\mathbb{F}_{q}$ satisfying $\sigma^2=\mathrm{id}_{\mathbb{F}_{q}}$. We say that $(\mathbb{F}_{q},\sigma)$ admits a \emph{$\tau$-monomial decomposition} if for any invertible self-adjoint matrix $A\in\mathbb{F}_{q}^{k\times k}$, there exists a unit lower triangular matrix $L\in\mathbb{F}_{q}^{k\times k}$ and a $\tau$-monomial matrix $DP_\tau\in\mathbb{F}_{q}^{k\times k}$ for some $\tau\in\mathrm{Sym}_{k}$ such that $LAL^{\star}=DP_\tau$.
\end{definition}

\begin{proposition}{\rm (\cite[Theorem 1]{CaoZhou2026})}\label{theorem-odd}
For a finite field $\mathbb{F}_{q}$ of odd characteristic and an automorphism $\sigma$ of $\mathbb{F}_{q}$ satisfying $\sigma^2={\rm id}_{\mathbb{F}_{q}}$,
there exists a $\tau$-monomial decomposition with respect to $(\mathbb{F}_{q},\sigma)$.
\end{proposition}

The following lemma is useful for establishing the main result of this subsection, namely Theorem \ref{theorem-com}.

\begin{lemma}\label{lem-nonexist}
If $\mathbb{F}_{q}$ has characteristic $2$ and $\sigma={\rm id}_{\mathbb{F}_{q}}$, then the $\tau$-monomial decomposition with respect to $(\mathbb{F}_{q},\sigma)$ does not exist.
\end{lemma}

\begin{IEEEproof}
Consider the following $2 \times 2$ invertible self-adjoint matrix over $\mathbb{F}_{q}$:
\begin{align*}
A =
\begin{pmatrix}
0 & 1 \\
1 & 1
\end{pmatrix}.
\end{align*}
Any $2\times 2$ unit lower triangular matrix $L$ over $\mathbb{F}_{q}$ has the form
\begin{align*}
L =
\begin{pmatrix}
1 & 0 \\
\alpha & 1
\end{pmatrix},
\end{align*}
where $\alpha \in \mathbb{F}_{q}$. Since the characteristic of $\mathbb{F}_{q}$ is $2$ and $\sigma={\rm id}_{\mathbb{F}_{q}}$, a direct calculation yields that
\begin{align*}
LAL^{\star}=A,
\end{align*}
which is not a $\tau$-monomial matrix for any $\tau\in{\rm Sym}_{2}$. Therefore, the $\tau$-monomial decomposition with respect to
$(\mathbb{F}_{q},\sigma)$ does not exist.
\end{IEEEproof}

\vspace{8pt}

In the following, we establish a unified $\tau$-monomial decomposition theorem, which completely characterizes the existence of the $\tau$-monomial decomposition with respect to $(\mathbb{F}_{q},\sigma)$ for finite fields $\mathbb{F}_{q}$ of odd characteristic or of characteristic $2$.

\begin{theorem}\textbf{\textup{(Unified $\tau$-Monomial Decomposition Theorem)}}\label{theorem-com}
Let $\sigma$ be an automorphism of $\mathbb{F}_{q}$ satisfying $\sigma^2=\mathrm{id}_{\mathbb{F}_{q}}$.
Then, there exists a $\tau$-monomial decomposition with respect to $(\mathbb{F}_{q},\sigma)$ if and only if one of following conditions holds:
\begin{itemize}
\item [(1)] $q$ is odd;

\item [(2)] $q=2^{2e}$ for some positive integer $e$, and $\sigma$ satisfies $\sigma(x)=x^{2^{e}}$ for all $x\in \mathbb{F}_{q}$.
\end{itemize}
\end{theorem}

\begin{IEEEproof}
By Proposition \ref{theorem-odd}, it suffices to consider the theorem in the case where $q$ is even.
Suppose that there exists a $\tau$-monomial decomposition with respect to $(\mathbb{F}_{q},\sigma)$, where $q$ is even.
By Lemma \ref{lem-nonexist}, we have $q=2^{2e}$ for some positive integer $e$, and $\sigma$ satisfies $\sigma(x)=x^{2^{e}}$ for all $x\in \mathbb{F}_{q}$.

Conversely, assume that $q=2^{2e}$ for some positive integer $e$, and $\sigma(x)=x^{2^{e}}$ for all $x\in \mathbb{F}_{q}$. We aim to show that there exists a $\tau$-monomial decomposition with respect to $(\mathbb{F}_{q},\sigma)$. We proceed by induction on the size $k\times k$ of the invertible self-adjoint matrix over
$\mathbb{F}_{q}$. The result holds trivially for $k=1$. Assume the conclusion holds for any invertible self-adjoint matrix in $\mathbb{F}_{q}^{k\times k}$ with $k<n$.
We shall show that it remains true for $k=n$.

Let $A=(a_{i,j})_{1\leq i,j\leq n}\in\mathbb{F}_{q}^{n\times n}$ be an arbitrary invertible self-adjoint matrix.
In what follows, we consider the cases $a_{1,1} \neq 0$ and $a_{1,1} = 0$, respectively.
\begin{enumerate}[wide, itemsep=0pt, leftmargin =0pt, widest={{\bf Case $1$}}]
\item[{\bf Case I}:] When $a_{1,1}\neq 0$, the matrix $A$ can be partitioned as
\begin{align*}
A =
\begin{pmatrix}
a_{1,1} & \mathbf{b} \\
\mathbf{b}^{\star} & A_{2,2}
\end{pmatrix},
\end{align*}
where $a_{1,1}=a_{1,1}^{\sqrt{q}}$ and $A_{2,2}=A_{2,2}^{\star}$. By setting
\begin{align*}
L_{1}=
\begin{pmatrix}
1&\mathbf{0}_{1\times (n-1)}\\
\frac{1}{a_{1,1}}\mathbf{b}^{\star}&I_{n-1}
\end{pmatrix},
\end{align*}
we derive
\begin{align*}
L_{1}AL_{1}^{\star}=
\begin{pmatrix}
a_{1,1}&\mathbf{0}_{1\times (n-1)}\\
\mathbf{0}_{(n-1)\times 1}&A_{2,2}+\frac{1}{a_{1,1}}\mathbf{b}^{\star}\mathbf{b}
\end{pmatrix}.
\end{align*}

The inductive hypothesis implies that there exists a unit lower triangular matrix $L_{2}$ and a $\tau_{1}$-monomial matrix $B_{1}$ for some $\tau_{1}\in\mathrm{Sym}_{n-1}$ such that
\begin{align*}
L_{2}\left(A_{2,2}+\frac{1}{a_{1,1}}\mathbf{b}^{\star}\mathbf{b}\right)L_{2}^{\star}=B_{1}.
\end{align*}

Taking the unit lower triangular matrix
\begin{align*}
L=
\begin{pmatrix}
1&\mathbf{0}_{1\times(n-1)}\\
\mathbf{0}_{(n-1)\times 1}&L_{2}\\
\end{pmatrix}{L}_{1},
\end{align*}
we derive
\begin{align*}
LAL^{\star}=
\begin{pmatrix}
a_{1,1}&\mathbf{0}_{1\times(n-1)}\\
\mathbf{0}_{(n-1)\times 1}&B_{1}\\
\end{pmatrix},
\end{align*}
which is a $\tau$-monomial matrix for some $\tau\in\mathrm{Sym}_{n}$.

\vspace{4pt}

\item[{\bf Case II}:] When $a_{1,1}=0$, we assume that $s$ is the smallest positive integer such that $a_{1,s}\neq 0$.
As $a_{i,j}=a_{j,i}^{\sqrt{q}}$ for all $1\leq i,j\leq n$, the matrix $A$ can be written as
\begin{align*}
A=
\begin{pmatrix}
0&\mathbf{0}_{1\times (s-2)}&a_{1,s}&\mathbf{b}\\
\mathbf{0}_{(s-2)\times 1}&A_{1,1}&\mathbf{c}^{\star}&A_{1,2}\\
a_{s,1}&\mathbf{c}&a_{s,s}&\mathbf{d}\\
\mathbf{b}^{\star}&A_{2,1}&\mathbf{d}^{\star}&A_{2,2}\\
\end{pmatrix}.
\end{align*}

As $a_{s,s}=a_{s,s}^{\sqrt{q}}$, it follows from Lemma \ref{cor-trace} that there exists $\beta\in \mathbb{F}_{q}$ such that $a_{s,s}=\beta^{\sqrt{q}}+\beta$.
We define a unit lower triangular matrix as follows:
\begin{align*}
L_{1}=
\begin{pmatrix}
1&\mathbf{0}_{1\times (s-2)}&0&\mathbf{0}_{1\times (n-s)}\\
\frac{1}{a_{1,s}}\mathbf{c}^{\star}&I_{s-2}&\mathbf{0}_{(s-2)\times 1}&O_{(s-2)\times (n-s)}\\
\frac{\beta}{a_{1,s}}&\mathbf{0}_{1\times (s-2)}&1&\mathbf{0}_{1\times (n-s)}\\
\frac{a_{s,s}}{a_{s,1}a_{1,s}}\mathbf{b}^{\star}+\frac{1}{a_{1,s}}\mathbf{d}^{\star}&O_{(n-s)\times (s-2)}&\frac{1}{a_{s,1}}\mathbf{b}^{\star}&I_{n-s}\\
\end{pmatrix}.
\end{align*}
Then, through direct computation and substituting $a_{s,s} = \beta^{\sqrt{q}} + \beta$, we obtain
\begin{align}\label{equation1}
L_{1}AL_{1}^{\star}=
\begin{pmatrix}
0&\mathbf{0}_{1\times (s-2)}&a_{1,s}&\mathbf{0}_{1\times (n-s)}\\
\mathbf{0}_{(s-2)\times 1}&A_{1,1}&\mathbf{0}_{(s-2)\times 1}&A_{1,2}+\frac{1}{a_{1,s}}\mathbf{c}^{\star}\mathbf{b}\\
a_{s,1}&\mathbf{0}_{1\times (s-2)}&0&\mathbf{0}_{1\times (n-s)}\\
\mathbf{0}_{(n-s)\times 1}&A_{2,1}+\frac{1}{a_{s,1}}\mathbf{b}^{\star}\mathbf{c}&\mathbf{0}_{(n-s)\times 1}
&A_{2,2}+\frac{1}{a_{1,s}}\mathbf{d}^{\star}\mathbf{b}+\frac{1}{a_{s,1}}\mathbf{b}^{\star}\mathbf{d}+\frac{a_{s,s}}{a_{s,1}a_{1,s}}\mathbf{b}^{\star}\mathbf{b}\\
\end{pmatrix}.
\end{align}

Let
\begin{align}\label{equation2}
A_{n-2}=
\begin{pmatrix}
A_{1,1}&A_{1,2}+\frac{1}{a_{1,s}}\mathbf{c}^{\star}\mathbf{b}\\
A_{2,1}+\frac{1}{a_{s,1}}\mathbf{b}^{\star}\mathbf{c}
&A_{2,2}+\frac{1}{a_{1,s}}\mathbf{d}^{\star}\mathbf{b}+\frac{1}{a_{s,1}}\mathbf{b}^{\star}\mathbf{d}
+\frac{a_{s,s}}{a_{s,1}a_{1,s}}\mathbf{b}^{\star}\mathbf{b}\\
\end{pmatrix}.
\end{align}
By the inductive hypothesis, there exists an $(n-2)\times (n-2)$ unit lower triangular matrix
\begin{align}\label{equation3}
L_{n-2}=
\begin{pmatrix}
L_{1,1}&O_{(s-2)\times (n-s)}\\
L_{2,1}&L_{2,2}\\
\end{pmatrix}
\end{align}
and a $\tau_{1}$-monomial matrix $B_{1}$ for some $\tau_{1}\in\mathrm{Sym}_{n-2}$ such that
\begin{align}\label{equation4}
L_{n-2}A_{n-2}L_{n-2}^{\star}=B_{1}.
\end{align}

Finally, we define an $n\times n$ unit lower triangular matrix
\begin{align}\label{equation15}
L=
\begin{pmatrix}
1&\mathbf{0}_{1\times (s-2)}&0&\mathbf{0}_{1\times (n-s)}\\
\mathbf{0}_{(s-2)\times 1}&L_{1,1}&\mathbf{0}_{(s-2)\times 1}&O_{(s-2)\times (n-s)}\\
0&\mathbf{0}_{1\times (s-2)}&1&\mathbf{0}_{1\times (n-s)}\\
\mathbf{0}_{(n-s)\times 1}&L_{2,1}&\mathbf{0}_{(n-s)\times 1}&L_{2,2}\\
\end{pmatrix}L_{1}.
\end{align}
Applying Eqs. \eqref{equation1}-\eqref{equation15}, it is straightforward to verify that $LAL^{\star}$ is a self-adjoint $\tau$-monomial matrix for some $\tau\in\mathrm{Sym}_{n}$.
\end{enumerate}

In summary, for both cases, there exists a unit lower triangular matrix $L$ such that $LAL^{\star}$ is a $\tau$-monomial matrix for some $\tau\in\mathrm{Sym}_{n}$. This confirms the existence of a $\tau$-monomial decomposition with respect to $(\mathbb{F}_q, \sigma)$, which completes the proof.
\end{IEEEproof}

\vspace{6pt}

\begin{theorem}\label{1theorem-com}
For any positive integer $k\leq q^{2}$, there exists a $k\times k$ $\tau$-OD matrix over $\mathbb{F}_{q^{2}}$ for some $\tau\in\mathrm{Sym}_{k}$.
\end{theorem}

\begin{IEEEproof}
By Proposition \ref{existNSC} and Theorem \ref{theorem-com}, we deduce the existence of a $\tau$-monomial decomposition with respect to $(\mathbb{F}_{q^2}, \sigma)$, regardless of whether $q$ is odd or even.

Note that there exists a unit lower triangular matrix $L\in\mathbb{F}_{q^{2}}^{k\times k}$ and a $\tau$-monomial matrix $DP_\tau\in\mathbb{F}_{q^{2}}^{k\times k}$ for some $\tau\in\mathrm{Sym}_{k}$ such that $LNN^{\dag}L^{\dag}=DP_\tau$. Moreover, it follows from Proposition \ref{LND} that $LN$ is an NSC matrix. Consequently, $LN$ is a $\tau$-OD matrix over $\mathbb{F}_{q^{2}}$ for some $\tau\in\mathrm{Sym}_{k}$, which confirms the statement.
\end{IEEEproof}

\subsection{Existence of Several New Infinite Families of $\tau$-OD Matrices}\label{OD}

Let $A$ be a $k \times k$ matrix. Let $\{u_i\}_{i=1}^s$ and $\{v_i\}_{i=1}^s$ be two sets of $s$ pairwise distinct positive integers in $\{1,2,\dots,k\}$.
We denote by
\begin{align*}
\left|\begin{pmatrix}
u_{1},&u_{2},&\cdots,&u_{s}\\
v_{1},&v_{2},&\cdots,&v_{s}
\end{pmatrix}_{A}\right|
\end{align*}
the determinant of the submatrix of $A$ formed by the $u_{1}',u_{2}',\ldots,u_{s}'$-th rows and the $v_{1}',v_{2}',\ldots,v_{s}'$-th columns, where $1\leq u_{1}'<u_{2}'<\ldots< u_{s}'\leq k$ and $1\leq v_{1}'<v_{2}'<\ldots< v_{s}'\leq k$, such that $\{u_{i}'\}_{i=1}^s=\{u_i\}_{i=1}^s$ and $\{v_{i}'\}_{i=1}^s=\{v_i\}_{i=1}^s$.

\vspace{3pt}

The following result corresponds to the finite field case of \cite[Lemma 3]{CaoZhou2026}.

\begin{lemma}\label{lem-permu}
Let $A=(a_{i,j})_{1\leq i,j\leq k}\in\mathbb{F}_{q^{2}}^{k\times k}$ be an invertible matrix. If there exists a unit lower triangular matrix $L\in\mathbb{F}_{q^{2}}^{k\times k}$ and a unit upper triangular matrix $U\in\mathbb{F}_{q^{2}}^{k\times k}$ such that $LAU$ is a $\tau$-monomial matrix for some $\tau\in{\rm Sym}_{k}$, then $\tau$ satisfies:
\begin{align*}
\tau(1)={\rm min}\left\{1\leq s \leq k:a_{1,s}\neq 0\right\},
\end{align*}
and for all $2\leq i\leq k$,
\begin{align*}
\tau(i)={\rm min}\left\{1\leq s \leq k:s\notin\{\tau(1),\ldots,\tau(i-1)\}, \
\left|\begin{pmatrix}
1, & \cdots, & i-1, & i \\
\tau(1), & \cdots , & \tau(i-1), & s
\end{pmatrix}_{A}\right|\neq 0\right\}.
\end{align*}
\end{lemma}

\vspace{6pt}

Building upon Theorem \ref{theorem-com} and Lemma \ref{lem-permu}, the following theorem establishes a unified approach for constructing $\tau$-OD matrices over finite fields of arbitrary characteristic, extending the scope of \cite[Theorem 4]{CaoZhou2026} from odd characteristic to all characteristics.

\begin{theorem}\label{theorem-permu}
Let $N\in\mathbb{F}_{q^{2}}^{k\times k}$ be an NSC matrix and $\tau\in{\rm Sym}_{k}$. Then, there exists a unit lower triangular matrix $L\in\mathbb{F}_{q^{2}}^{k\times k}$ such that $LN$ is a $\tau$-OD matrix over $\mathbb{F}_{q^{2}}$ if and only if
\begin{align}\label{equation20}
\tau(1)={\rm min}\left\{1\leq s \leq k:n_{1,s}\neq 0\right\},
\end{align}
where $n_{1,s}$ is the $(1,s)$-entry of $NN^{\dag}$, and for all $2\leq i\leq k$,
\begin{align}\label{equation21}
\tau(i)={\rm min}\left\{1\leq s \leq k:s\notin\{\tau(1),\ldots,\tau(i-1)\}, \
\left|\begin{pmatrix}
1, & \cdots, & i-1, & i \\
\tau(1), & \cdots , & \tau(i-1), & s
\end{pmatrix}_{NN^{\dag}}\right|\neq 0\right\}.
\end{align}
\end{theorem}

\begin{IEEEproof}
``$\Longrightarrow$'': Suppose there exists a unit lower triangular matrix $L\in\mathbb{F}_{q^{2}}^{k\times k}$ such that $LN$ is a $\tau$-OD matrix over $\mathbb{F}_{q^{2}}$. Then, $LN(LN)^{\dag}$ is a $\tau$-monomial matrix.
Setting $A=NN^{\dag}$ and $U=L^{\dag}$ in Lemma \ref{lem-permu}, it follows that Eqs. \eqref{equation20} and \eqref{equation21} hold.

\vspace{4pt}

``$\Longleftarrow$'': Conversely, suppose $\tau$ satisfies Eqs. \eqref{equation20} and \eqref{equation21}.
By Theorem \ref{theorem-com}, there exists a $\tau$-monomial decomposition with respect to $(\mathbb{F}_{q^2}, \sigma)$, regardless of whether $q$ is odd or even.
Since $NN^{\dag}$ is an invertible self-adjoint matrix, there exists a unit lower triangular matrix $L\in\mathbb{F}_{q^{2}}^{k\times k}$ such that
$LNN^{\dag}L^{\dag}$ is a $\tau'$-monomial matrix for some $\tau'\in\mathrm{Sym}_{k}$. By Proposition \ref{LND} and Lemma \ref{lem-permu}, $\tau'=\tau$, and therefore $LN$ is a $\tau$-OD matrix over $\mathbb{F}_{q^{2}}$.

Therefore, we complete the proof.
\end{IEEEproof}

\vspace{6pt}

Now, by applying Theorem \ref{theorem-permu}, we prove the existence of a new infinite family of $\tau$-OD matrices over finite fields of characteristic $2$ in the following theorem.

\begin{theorem}\label{theorem-inf1}
Let $q=2^e>2$ be an even prime power. Suppose $(k-1)\mid(q^{2}-1)$ with $k\geq 3$. Let $i+qt_{i}\equiv 0 \pmod{k-1}$, where $1\leq t_{i}\leq k-2$ for $1\leq i\leq k-2$. Then, there exists an infinite family of $k\times k$ $\tau$-OD matrices over $\mathbb{F}_{q^2}$, where
\begin{align}\label{equ-perm}
\tau=\begin{pmatrix}
1 & 2 & 3 & \cdots & k-1 & k\\
1 & t_{1}+1 & t_{2}+1 & \cdots & t_{k-2}+1 & k
\end{pmatrix}.
\end{align}
\end{theorem}

\begin{IEEEproof}
We define the following $k\times k$ matrix over $\mathbb{F}_{q^2}$:
\begin{align*}
N=\begin{pmatrix}
g&1&1&1&\cdots&1\\
0&1&g^{m}&g^{2m}&\cdots&g^{(k-2)m}\\
0&1&g^{2m}&g^{4m}&\cdots&g^{2(k-2)m}\\
\vdots&\vdots&\vdots&\vdots&\ddots&\vdots\\
0&1&g^{(k-2)m}&g^{2(k-2)m}&\cdots&g^{(k-2)^{2}m}\\
0&1&1&1&\cdots&1\\
\end{pmatrix},
\end{align*}
where $g$ is a primitive element of $\mathbb{F}_{q^2}$ and $m=\frac{q^2-1}{k-1}$. One can easily verify that $N$ is an NSC matrix. To compute $NN^{\dag}$, let $\mathbf{n}_{i}$ denote the $i$-th row of $N$ for $1\leq i\leq k$. The following properties regarding the Hermitian inner products of these rows are then obtained.

\begin{itemize}
\item [1)] Clearly, $\langle\mathbf{n}_{1},\mathbf{n}_{1}\rangle_{\mathrm{H}}=g^{q+1}+1$ and
$\langle\mathbf{n}_{1},\mathbf{n}_{k}\rangle_{\mathrm{H}}=\langle\mathbf{n}_{k},\mathbf{n}_{1}\rangle_{\mathrm{H}}
=\langle\mathbf{n}_{k},\mathbf{n}_{k}\rangle_{\mathrm{H}}=1$.

\vspace{6pt}

\item [2)] When $2\leq i\leq k-1$, we have
\begin{align*}
\langle\mathbf{n}_{1},\mathbf{n}_{i}\rangle_{\mathrm{H}}=\langle\mathbf{n}_{k},\mathbf{n}_{i}\rangle_{\mathrm{H}}=0.
\end{align*}
The fact $(NN^{\dag})^{\dag}=NN^{\dag}$ implies that for $2\leq i\leq k-1$,  $\langle\mathbf{n}_{i},\mathbf{n}_{1}\rangle_{\mathrm{H}}=\langle\mathbf{n}_{i},\mathbf{n}_{k}\rangle_{\mathrm{H}}=0$.

\vspace{6pt}

\item [3)] When $2\leq i,j\leq k-1$, we obtain
\begin{align*}
\langle\mathbf{n}_{i},\mathbf{n}_{j}\rangle_{\mathrm{H}}
&=\sum_{s=0}^{k-2}g^{sm[i-1+q(j-1)]} \\
&=\begin{cases}
k-1, & \mbox{if } j=t_{i-1}+1,\\
0,& \mbox{if } j\neq t_{i-1}+1.
\end{cases}
\end{align*}
\end{itemize}

It follows that
\begin{align*}
NN^{\dag}=(g^{q+1}+1)E_{1,1}+(k-1)\left(E_{1,k}+E_{k,1}+E_{k,k}+\sum_{i=2}^{k-1}E_{i,t_{i-1}+1}\right).
\end{align*}
Here $E_{i,j}$ denotes the matrix whose $(i,j)$-entry is $1$, while all other entries are $0$.

By Theorem \ref{theorem-com}, there exists a unit lower triangular matrix $L$ such that $LNN^{\dag}L^{\dag}$ is a $\tau$-monomial matrix
over $\mathbb{F}_{q^{2}}$ for some $\tau\in{\rm Sym}_{k}$. Hence, $LN$ is a $\tau$-OD matrix. Applying Theorem \ref{theorem-permu}, we derive $\tau(1)=1$, $\tau(k)=k$, and
$\tau(i)=t_{i-1}+1$ for $2\leq i\leq k-1$. This result is consistent with Eq. \eqref{equ-perm}, which completes the proof.
\end{IEEEproof}

\vspace{6pt}

In the following theorem, we demonstrate the existence of several infinite families of small-sized $\tau$-OD matrices over finite fields of characteristic $2$.

\begin{theorem}\label{theorem-inf2}
Let $e$ be a positive integer. The following three statements hold.
\begin{itemize}
\item [(1)] There exists an infinite family of $2 \times 2$ $\tau$-OD matrices over $\mathbb{F}_{2^{2e}}$ with $\tau=(1,2)$.

\item [(2)] If $e > 1$, then there exists an infinite family of $3 \times 3$ $\tau$-OD matrices over $\mathbb{F}_{2^{2e}}$ with $\tau=(2,3)$.

\item [(3)] If $e > 2$, then there exist an infinite family of $4 \times 4$ $\tau$-OD matrices over $\mathbb{F}_{2^{2e}}$ with $\tau=(1,2)(3,4)$.
\end{itemize}
\end{theorem}

\begin{IEEEproof}
(1) Let
\begin{align*}
N=\begin{pmatrix}
1&1\\
0&1\\
\end{pmatrix}.
\end{align*}
Then, $N$ is an NSC matrix and
\begin{align*}
NN^{\dag}=\begin{pmatrix}
0&1\\
1&1\\
\end{pmatrix}.
\end{align*}

In light of Theorem \ref{theorem-com}, there exists a unit lower triangular matrix $L$ such that $LNN^{\dag}L^{\dag}$ is a $\tau$-monomial matrix
over $\mathbb{F}_{2^{2e}}$ for some $\tau\in{\rm Sym}_{2}$. Hence, $LN$ is a $\tau$-OD matrix. By Theorem \ref{theorem-permu}, we derive $\tau=(1,2)$. Therefore, $LN$ is a $2 \times 2$ $\tau$-OD matrix with $\tau=(1,2)$, which completes the proof of statement (1).

\vspace{6pt}

(2) When $e > 1$, let $\alpha=g^{2^{e}+1}$, where $g$ is a primitive element of $\mathbb{F}_{2^{2e}}$. Let
\begin{align*}
N=\begin{pmatrix}
1&1&1\\
0&1&\alpha\\
0&1&1\\
\end{pmatrix}.
\end{align*}
Then, $N$ is an NSC matrix. We derive
\begin{align*}
NN^{\dag}=\begin{pmatrix}
1&1+\alpha&0\\
1+\alpha&1+\alpha^{2}&1+\alpha\\
0&1+\alpha&0\\
\end{pmatrix}.
\end{align*}

By Theorem \ref{theorem-com}, there exists a unit lower triangular matrix $L$ such that $LNN^{\dag}L^{\dag}$ is a $\tau$-monomial matrix
over $\mathbb{F}_{2^{2e}}$ for some $\tau\in{\rm Sym}_{3}$. Hence, $LN$ is a $\tau$-OD matrix.
Next, we apply Theorem \ref{theorem-permu} to compute the permutation $\tau$. Since the $(1,1)$-entry of $NN^{\dag}$ is nonzero, we obtain $\tau(1)=1$.
We have
\begin{align*}
\left|\begin{pmatrix}
1, & 2 \\
1, & 2
\end{pmatrix}_{NN^{\dag}}\right|
=0, \ \
\left|\begin{pmatrix}
1, & 2 \\
1, & 3
\end{pmatrix}_{NN^{\dag}}\right|
\neq 0.
\end{align*}
This implies that $\tau=(2,3)$.
Therefore, $LN$ is a $3 \times 3$ $\tau$-OD matrix with $\tau=(2,3)$, which completes the proof of statement (2).

\vspace{6pt}

(3) When $e>2$, let $\alpha=g^{2^{e}+1}$, where $g$ is a primitive element of $\mathbb{F}_{2^{2e}}$. Let
\begin{align*}
N=\begin{pmatrix}
1&1&1&1\\
0&1&\alpha&\alpha^2\\
0&1&0&1\\
0&0&0&1\\
\end{pmatrix}.
\end{align*}
Then, $N$ is an NSC matrix. We derive
\begin{align*}
NN^{\dag}=\begin{pmatrix}
0&1+\alpha+\alpha^{2}&0&1\\
1+\alpha+\alpha^{2}&1+\alpha^{2}+\alpha^{4}&1+\alpha^2&\alpha^2\\
0&1+\alpha^{2}&0&1\\
1&\alpha^{2}&1&1
\end{pmatrix}.
\end{align*}

By Theorem \ref{theorem-com}, there exists a unit lower triangular matrix $L$ such that $LNN^{\dag}L^{\dag}$ is a $\tau$-monomial matrix
over $\mathbb{F}_{2^{2e}}$ for some $\tau\in{\rm Sym}_{4}$. Hence, $LN$ is a $\tau$-OD matrix.
Let us employ Theorem \ref{theorem-permu} to compute the permutation $\tau$. Since the $(1,2)$-entry of $NN^{\dag}$ is $1+\alpha+\alpha^{2}\neq 0$,
we have $\tau(1)=2$. Similarly, $\tau(2)=1$. Since
\begin{align*}
\left|\begin{pmatrix}
1, & 2, &3 \\
2, & 1, &3
\end{pmatrix}_{NN^{\dag}}\right|
=0,\
\left|\begin{pmatrix}
1, & 2, &3 \\
2, & 1, &4
\end{pmatrix}_{NN^{\dag}}\right|
\neq 0,
\end{align*}
we obtain $\tau=(1,2)(3,4)$. Therefore, $LN$ is a $4 \times 4$ $\tau$-OD matrix with $\tau=(1,2)(3,4)$,
which completes the proof of statement (3).
\end{IEEEproof}

\begin{remark}\label{order2}
We make the following remarks regarding Theorem \ref{theorem-inf2}.
\begin{itemize}
\item [(1)] By combining Theorem \ref{theorem-inf2} (1) with the result in \cite[Theorem 6]{CaoZhou2026}, we conclude that $2 \times 2$ $\tau$-OD matrices over
$\mathbb{F}_{q^{2}}$ with $\tau=(1,2)$ always exist for any prime power $q$.

\item [(2)] By combining Theorem \ref{theorem-inf2} (2) with the result in \cite[Theorem 7]{CaoZhou2026}, we conclude that $3 \times 3$ $\tau$-OD matrices over
$\mathbb{F}_{q^{2}}$ with $\tau=(2,3)$ always exist for any prime power $q>2$.

\item [(3)] In \cite[Theorem 8]{CaoZhou2026}, it was shown that $4 \times 4$ $\tau$-OD matrices over $\mathbb{F}_{q^{2}}$ with $\tau=(1,4)$ always exist for any odd
prime power $q$. In contrast, Theorem \ref{theorem-inf2} (3) verifies the existence of an infinite family of $4 \times 4$ $\tau$-OD matrices over $\mathbb{F}_{2^{2e}}$ with a distinct permutation $\tau=(1,2)(3,4)$.
\end{itemize}
\end{remark}

\section{Several Infinite Families of Quantum Codes and Some Record-Breaking Quantum Codes}\label{infinitequ}

Based on the $\tau$-OD matrices from the previous section, this section aims to construct quantum codes with flexible parameters and to obtain some new record-breaking quantum codes.

\subsection{Several Infinite Families of Quantum Codes with Flexible Parameters}\label{twothreefour}

For any row vector $\mathbf{x}=(x_{1},x_{2},\ldots,x_{n})$, we define $\mathbf{x}^{q}:=(x_{1}^{q},x_{2}^{q},\ldots,x_{n}^{q})$.
Utilizing the $\tau$-OD matrices established in Theorem \ref{theorem-inf1}, we derive an infinite family of quantum codes as follows.

\begin{theorem}\label{theorem-qcode1}
Let $q=2^{e}>2$ be an even prime power. Suppose $(k-1)\mid(q^{2}-1)$, with $k\geq 3$ and $\mathrm{gcd}(k-1,q+1)=1$.
Let $i+qt_{i}\equiv 0 \pmod{k-1}$, where $1\leq t_{i}\leq k-2$ for $1\leq i\leq k-2$.
Put $n=mq-r$ with $1 \leq m\leq q$ and $0\leq r \leq q-1$. Then, there exists an infinite family of quantum codes with parameters
\begin{align}\label{lengthkn}
\bigg[\mspace{-4mu}\bigg[kn,2\sum_{i=1}^{k}r_{i}-kn,\geq \min\{(k-i+1)(n-r_{i}+1):1\leq i\leq k\}\bigg]\mspace{-4mu}\bigg]_{q},
\end{align}
where $1\leq r_{i}\leq n$ for $1\leq i\leq k$, $n-(q-1-\lfloor r/m\rfloor)/2\leq r_1,r_{k}\leq n-2$, and $r_{j+1}+r_{t_{j}+1}\geq n$ for $1\leq j\leq k-2$.
\end{theorem}

\begin{IEEEproof}
Since $n=mq-r$ with $1 \leq m\leq q$ and $0\leq r \leq q-1$, it follows from \cite[Theorem 3.4]{Jin2010Application} that there exist Hermitian dual-containing GRS codes $\mathcal{C}_1$ and $\mathcal{C}_k$ with parameters $[n,r_1,n-r_1+1]_{q^2}$ and $[n,r_k,n-r_k+1]_{q^2}$, respectively, where $n-(q-1-\lfloor r/m\rfloor)/2\leq r_1,r_k\leq n-2$.

Since $1\leq r_{t_{j}+1}\leq n\leq q^2$, there exists a GRS code $\mathcal{C}_{t_{j}+1}=\mathrm{GRS}_{r_{t_{j}+1}}(\mathbf{x},\mathbf{y})$ with parameters $[n,r_{t_{j}+1},n-r_{t_{j}+1}+1]_{q^2}$ for some vectors $\mathbf{x}$ and $\mathbf{y}$ in $\mathbb{F}_{q^{2}}^{n}$.
Clearly, $\mathcal{C}_{t_{j}+1}^{\bot_{\mathrm{H}}}=\mathrm{GRS}_{n-r_{t_{j}+1}}(\mathbf{x}^{q},\mathbf{z})$ is also a GRS code whose parameters are $[n,n-r_{t_{j}+1},r_{t_{j}+1}+1]_{q^2}$ for some vector $\mathbf{z}$ in $\mathbb{F}_{q^{2}}^{n}$. We set $\mathcal{C}_{j+1}=\mathrm{GRS}_{r_{j+1}}(\mathbf{x}^{q},\mathbf{z})$.
Then, $\mathcal{C}_{t_{j}+1}^{\bot_{\mathrm{H}}}\subseteq \mathcal{C}_{j+1}$ for $1\leq j\leq k-2$.

According to Theorem \ref{theorem-inf1}, there exist $k\times k$ $\tau$-OD matrices over $\mathbb{F}_{q^{2}}$ for any even prime power $q=2^e>2$, where $\tau$ is given in Eq. \eqref{equ-perm}.
Define an MP code $\mathcal{C}(A)=[\mathcal{C}_{1},\mathcal{C}_{2},\ldots,\mathcal{C}_{k}]\cdot A$.
It follows from Lemma \ref{proposition4} and Eq. \eqref{equ-perm} that
\begin{align*}
\mathcal{C}(A)^{\bot_{\mathrm{H}}}\subseteq\mathcal{C}(A).
\end{align*}

Consequently, $\mathcal{C}(A)$ is a Hermitian dual-containing MP code, and it has parameters $[kn,\sum_{i=1}^k r_i,\geq \min\{(k-i+1)(n-r_i+1):1\leq i\leq k\}]_{q^{2}}$.
By Lemma \ref{proposition6}, $\mathcal{C}(A)$ yields a quantum code with parameters as presented in Eq. \eqref{lengthkn}.

This completes the proof.
\end{IEEEproof}

\vspace{6pt}

Next, we employ $2\times 2$, $3\times 3$, and $4\times 4$ $\tau$-OD matrices to derive several infinite families of quantum codes with flexible parameters, as presented in Theorems \ref{theorem-qcode2}, \ref{theorem-qcode3} and \ref{theorem-qcode4}.

\begin{theorem}\label{theorem-qcode2}
Let $q$ be a prime power. Suppose $2\leq n\leq q^2$, $1\leq r_1,r_2\leq n$ and $r_1+r_2\geq n$. Then, there exists an infinite family of quantum codes with parameters
\begin{align}\label{2n}
\bigg[\mspace{-4mu}\bigg[2n,2\sum_{i=1}^2 r_i-2n,\geq \min\{(3-i)(n-r_i+1):1\leq i\leq 2\}\bigg]\mspace{-4mu}\bigg]_{q}.
\end{align}
\end{theorem}

\begin{IEEEproof}
Since $2\leq n\leq q^2$ and $1\leq r_2\leq n$, there exists a GRS code $\mathcal{C}_2=\mathrm{GRS}_{r_{2}}(\mathbf{u},\mathbf{v})$ with parameters $[n,r_2,n-r_2+1]_{q^2}$ for some vectors $\mathbf{u}$ and $\mathbf{v}$ in $\mathbb{F}_{q^{2}}^{n}$. Clearly,
$\mathcal{C}_{2}^{\bot_{\mathrm{H}}}=\mathrm{GRS}_{n-r_{2}}(\mathbf{u}^{q},\mathbf{w})$ is also a GRS code whose parameters are $[n,n-r_2,r_2+1]_{q^2}$ for some vector
$\mathbf{w}$ in $\mathbb{F}_{q^{2}}^{n}$. We set $\mathcal{C}_1=\mathrm{GRS}_{r_{1}}(\mathbf{u}^{q},\mathbf{w})$.
Then, $\mathcal{C}_{2}^{\bot_{\mathrm{H}}}\subseteq \mathcal{C}_{1}$.

According to Remark \ref{order2}, there exist $2\times 2$ $\tau$-OD matrices over $\mathbb{F}_{q^{2}}$ with $\tau=(1,2)$ for any prime power $q$.
Define an MP code $\mathcal{C}(A)=[\mathcal{C}_{1},\mathcal{C}_{2}]\cdot A$. Similar to the proof of Theorem \ref{theorem-qcode1}, we derive
\begin{align*}
\mathcal{C}(A)^{\bot_{\mathrm{H}}}\subseteq\mathcal{C}(A).
\end{align*}

Consequently, $\mathcal{C}(A)$ is a Hermitian dual-containing MP code with parameters $[2n,\sum_{i=1}^2 r_i,\geq \min\{(3-i)(n-r_i+1):1\leq i\leq 2\}]_{q^{2}}$.
By Lemma \ref{proposition6}, $\mathcal{C}(A)$ yields a quantum code with parameters as presented in Eq. \eqref{2n}.

This completes the proof.
\end{IEEEproof}

\begin{theorem}\label{theorem-qcode3}
Let $q>2$ be a prime power. Put $n=mq-r$ with $1 \leq m\leq q$ and $0\leq r \leq q-1$. Suppose $n-(q-1-\lfloor r/m\rfloor)/2\leq r_1\leq n-2$,
$1\leq r_2,r_3\leq n$ and $r_2+r_3\geq n$. Then, there exists an infinite family of quantum codes with parameters
\begin{align}\label{3n}
\bigg[\mspace{-4mu}\bigg[3n,2\sum_{i=1}^3 r_i-3n,\geq \min\{(4-i)(n-r_i+1):1\leq i\leq 3\}\bigg]\mspace{-4mu}\bigg]_{q},
\end{align}
\end{theorem}

\begin{IEEEproof}
Since $n=mq-r$ with $1 \leq m\leq q$ and $0\leq r \leq q-1$, it follows from \cite[Theorem 3.4]{Jin2010Application} that there exists a Hermitian dual-containing GRS code $\mathcal{C}_1=\mathrm{GRS}_{r_{1}}(\mathbf{u},\mathbf{v})$ with parameters $[n,r_1,n-r_1+1]_{q^2}$ for some vectors $\mathbf{u}$ and $\mathbf{v}$ in $\mathbb{F}_{q^{2}}^{n}$, where $n-(q-1-\lfloor r/m\rfloor)/2\leq r_1\leq n-2$.
Since $n\leq q^2$ and $1\leq r_3\leq n$, there exists a GRS code $\mathcal{C}_3=\mathrm{GRS}_{r_{3}}(\mathbf{x},\mathbf{y})$ with parameters $[n,r_3,n-r_3+1]_{q^2}$
for some vectors $\mathbf{x}$ and $\mathbf{y}$.
Clearly, $\mathcal{C}_{3}^{\bot_{\mathrm{H}}}=\mathrm{GRS}_{n-r_{3}}(\mathbf{x}^{q},\mathbf{z})$ is also a GRS code whose parameters are $[n,n-r_3,r_3+1]_{q^2}$ for some vector $\mathbf{z}$. We set $\mathcal{C}_2=\mathrm{GRS}_{r_{2}}(\mathbf{x}^{q},\mathbf{z})$.
Then, $\mathcal{C}_{3}^{\bot_{\mathrm{H}}}\subseteq \mathcal{C}_{2}$.

According to Remark \ref{order2}, there exist $3\times 3$ $\tau$-OD matrices over $\mathbb{F}_{q^{2}}$ with $\tau=(2,3)$ for any prime power $q>2$.
Define an MP code $\mathcal{C}(A)=[\mathcal{C}_{1},\mathcal{C}_{2},\mathcal{C}_{3}]\cdot A$. Similar to the proof of Theorem \ref{theorem-qcode1}, we derive
\begin{align*}
\mathcal{C}(A)^{\bot_{\mathrm{H}}}\subseteq\mathcal{C}(A).
\end{align*}

Consequently, $\mathcal{C}(A)$ is a Hermitian dual-containing MP code, and it has parameters $[3n,\sum_{i=1}^3 r_i,\geq \min\{(4-i)(n-r_i+1):1\leq i\leq 3\}]_{q^{2}}$.
By Lemma \ref{proposition6}, $\mathcal{C}(A)$ yields a quantum code with parameters as presented in Eq. \eqref{3n}.

This completes the proof.
\end{IEEEproof}

\begin{theorem}\label{theorem-qcode4}
Let $q>4$ be an even prime power. Suppose $1\leq r_1,r_2,r_3,r_4\leq n$, $r_1+r_2\geq n$ and $r_3+r_4\geq n$. Then, there exists an infinite family of quantum codes with parameters
\begin{align}\label{4n}
\bigg[\mspace{-4mu}\bigg[4n,2\sum_{i=1}^4 r_i-4n,\geq \min\{(5-i)(n-r_i+1):1\leq i\leq 4\}\bigg]\mspace{-4mu}\bigg]_{q},
\end{align}
\end{theorem}

\begin{IEEEproof}
Similar to the proof of Theorem \ref{theorem-qcode1}, under the conditions $1\leq r_1,r_2,r_3,r_4\leq n$, $r_1+r_2\geq n$ and $r_3+r_4\geq n$,
we can find GRS codes $\mathcal{C}_{i}$ with parameters $[n,r_i,n-r_i+1]_{q^2}$ for $1\leq i \leq 4$ such that
$\mathcal{C}_{2}^{\bot_{\mathrm{H}}}\subseteq \mathcal{C}_{1}$ and $\mathcal{C}_{4}^{\bot_{\mathrm{H}}}\subseteq \mathcal{C}_{3}$.

According to Theorem \ref{theorem-inf2} (3), there exist $4\times 4$ $\tau$-OD matrices over $\mathbb{F}_{q^{2}}$ with $\tau=(1,2)(3,4)$ for any even prime power $q>4$. Define an MP code $\mathcal{C}(A)=[\mathcal{C}_{1},\mathcal{C}_{2},\mathcal{C}_{3},\mathcal{C}_{4}]\cdot A$. Similar to the proof of Theorem \ref{theorem-qcode1}, we derive
\begin{align*}
\mathcal{C}(A)^{\bot_{\mathrm{H}}}\subseteq\mathcal{C}(A).
\end{align*}

Consequently, $\mathcal{C}(A)$ is a Hermitian dual-containing MP code, and it has parameters $[4n,\sum_{i=1}^4 r_i,\geq \min\{(5-i)(n-r_i+1):1\leq i\leq 4\}]_{q^{2}}$.
By Lemma \ref{proposition6}, $\mathcal{C}(A)$ yields a quantum code with parameters as presented in Eq. \eqref{4n}.

This completes the proof.
\end{IEEEproof}

\subsection{Examples of Record-Breaking Quantum Codes}

In this subsection, we first provide several representative examples of $q$-ary record-breaking quantum codes for $q=4,5,7,8$ derived from the constructions
in Subsection \ref{twothreefour}. Subsequently, we present a total of $222$ record-breaking quantum codes in Tables \ref{table1} and \ref{table2}.

\begin{example}\label{example1-x}
Consider $q=4$, $n=15$, $r_1=11$ and $r_2=7$ in Theorem \ref{theorem-qcode2}. Then,
\begin{align*}
\min\{(3-i)(n-r_i+1):1\leq i\leq 2\}=9.
\end{align*}
By Theorem \ref{theorem-qcode2}, we obtain a $[[30,6,\geq9]]_{4}$ quantum code, which matches the best-known parameters in Grassl's database \cite{Grassl2025Bounds}. Furthermore, the generator matrix $G$ of the MP code in Theorem \ref{theorem-qcode2} can be taken as
\begin{align*}
G=\begin{pmatrix}
G_1 & G_1\\
G_2& g^3G_2
\end{pmatrix},
\end{align*}
where $g$ is a primitive element of $\mathbb{F}_{4^2}$, and $G_1$ and $G_2$ are respectively given by
\begin{align*}
G_1=\begin{pmatrix}
1 & 1 & 1&\cdots& 1\\
1 & g & g^2& \cdots& g^{14}\\
\vdots & \vdots & \vdots& \ddots& \vdots\\
1 & g^{10} & g^{2\times 10}& \cdots& g^{14\times 10}
\end{pmatrix}\ \mathrm{and} \  G_2=\begin{pmatrix}
1 & g^{4} & g^{4\times 2}& \cdots& g^{4\times 14}\\
1 & g^{8} & g^{8\times 2}& \cdots& g^{8\times 14}\\
\vdots & \vdots & \vdots& \ddots& \vdots\\
1 & g^{28} & g^{28\times 2}& \cdots& g^{28\times 14}
\end{pmatrix}.
\end{align*}

After verification via the Magma algebra software \cite{Bosma1997The}, this MP code is confirmed to be a Hermitian dual-containing code with parameters
$[30,18,10]_{4^2}$. Hence, it produces a record-breaking quantum code with parameters
\begin{align*}
[[30,6,\geq\textbf{10}]]_{4},
\end{align*}
which improves the minimum distance of the best-known $[[30,6,9]]_{4}$ quantum code maintained in Grassl's database \cite{Grassl2025Bounds}.
Applying the lengthening construction, it produces a record-breaking $[[31,6,\geq\textbf{10}]]_{4}$ quantum code, which improves the best-known $[[31,6,9]]_{4}$ quantum code maintained in \cite{Grassl2025Bounds}.
Utilizing the subcode construction, it also yields a record-breaking $[[30,5,\geq\textbf{10}]]_{4}$ quantum code, which improves the best-known $[[30,5,9]]_{4}$ quantum code maintained in \cite{Grassl2025Bounds}.
\end{example}

\begin{example}
Consider $q=5$, $n=25$, $r_1=18$ and $r_2=11$ in Theorem \ref{theorem-qcode2}. Then,
\begin{align*}
\min\{(3-i)(n-r_i+1):1\leq i\leq 2\}=15.
\end{align*}
By Theorem \ref{theorem-qcode2}, we obtain a $[[50,8,\geq\textbf{15}]]_{5}$ quantum code, which improves the minimum distance of the best-known $[[50,8,14]]_{5}$
quantum code maintained in Grassl's database \cite{Grassl2025Bounds}.
Applying the lengthening construction, it produces a record-breaking $[[51,8,\geq\textbf{15}]]_{5}$ quantum code, which improves the best-known $[[51,8,14]]_{5}$
quantum code maintained in \cite{Grassl2025Bounds}.
Utilizing the subcode construction, it also yields a record-breaking $[[50,7,\geq\textbf{15}]]_{5}$ quantum code, which improves the best-known $[[50,7,14]]_{5}$
quantum code maintained in \cite{Grassl2025Bounds}.
\end{example}

\begin{example}
Consider $q=7$, $n=39$, $r_1=34$ and $r_2=28$ in Theorem \ref{theorem-qcode2}. Then,
\begin{align*}
\min\{(3-i)(n-r_i+1):1\leq i\leq 2\}=12.
\end{align*}
By Theorem \ref{theorem-qcode2}, we obtain a $[[78,46,\geq\textbf{12}]]_{7}$ quantum code, which improves the minimum distance of the best-known $[[78,46,10]]_{7}$ quantum code maintained in Grassl's database \cite{Grassl2025Bounds}.
Applying the lengthening construction, it produces a record-breaking $[[79,46,\geq\textbf{12}]]_{7}$ quantum code, which improves the best-known $[[79,46,10]]_{7}$ quantum code maintained in \cite{Grassl2025Bounds}.
Utilizing the subcode construction, it also yields a record-breaking $[[78,45,\geq\textbf{12}]]_{7}$ quantum code, which improves the best-known $[[78,45,10]]_{7}$ quantum code maintained in \cite{Grassl2025Bounds}.
\end{example}

\begin{example}\label{example2-x}
Consider $q=8$, $m=5$ and $r=7$ in Theorem \ref{theorem-qcode3}. Then, $n=33$. Setting $r_1=31$, $r_{2}=30$ and $r_{3}=26$, we have
\begin{align*}
\min\{(4-i)(n-r_i+1):1\leq i\leq 3\}=8.
\end{align*}
By Theorem \ref{theorem-qcode3}, we obtain a $[[99,75,\geq\textbf{8}]]_{8}$ quantum code, which improves the minimum distance of the best-known $[[99,75,7]]_{8}$ quantum code maintained in Grassl's database \cite{Grassl2025Bounds}.
Applying the lengthening construction, it produces a record-breaking $[[100,75,\geq\textbf{8}]]_{8}$ quantum code, which improves the best-known $[[100,75,7]]_{8}$ quantum code maintained in \cite{Grassl2025Bounds}.
Utilizing the subcode construction, it also yields a record-breaking $[[99,74,\geq\textbf{8}]]_{8}$ quantum code, which improves the best-known $[[99,74,7]]_{8}$ quantum code maintained in \cite{Grassl2025Bounds}.
\end{example}

\begin{remark}
In addition to the quantum codes shown in Examples \ref{example1-x}-\ref{example2-x}, Tables \ref{table1} and \ref{table2} exhibit a total of $240$ quantum codes,
which are compared with the best-known records maintained in Grassl's database \cite{Grassl2025Bounds}.
We highlight that $222$ of these quantum codes break the best-known records maintained in \cite{Grassl2025Bounds}. Specifically,
\begin{itemize}
\item [(1)] In Table \ref{table1}, we provide $80$ quantum codes derived from Theorems \ref{theorem-qcode2} and \ref{theorem-qcode3},
all of which improve the best-known lower bounds on the minimum distances of quantum codes maintained in Grassl's database \cite{Grassl2025Bounds}.

\vspace{4pt}

\item [(2)] In Table \ref{table2}, we present $160$ quantum codes by applying the lengthening and subcode constructions to the $80$ quantum
codes in Table \ref{table1}. As shown in Table \ref{table2}, $142$ of these quantum codes (whose minimum distance lower bounds are marked in bold) outperform the best-known records maintained in Grassl's database \cite{Grassl2025Bounds}.
\end{itemize}
\end{remark}

\begin{table}[!htbp]
\renewcommand\arraystretch{1.3}
\centering	
\footnotesize
\setlength{\abovecaptionskip}{0.cm}
\setlength{\belowcaptionskip}{0.1cm}
\caption{Comparison of our $q$-ary quantum codes with those in Grassl's database \cite{Grassl2025Bounds} for $q=4,5,7,8$}\label{table1}
\vspace{3pt}
\begin{tabular}{cccc|cccc}
\hline
No.&Our quantum code&Reference&Record in \cite{Grassl2025Bounds}&No.&Our quantum code&Reference&Record in \cite{Grassl2025Bounds}\\
\hline
$1$&$[[30,6,\geq\textbf{10}]]_{4}$&Theorem \ref{theorem-qcode2}&$[[30,6,9]]_{4}$&$41$&$[[81,57,\geq\textbf{8}]]_{8}$&Theorem \ref{theorem-qcode3}&$[[81,57,7]]_{8}$\\

$2$&$[[44,18,\geq\textbf{10}]]_{5}$&Theorem \ref{theorem-qcode2}&$[[44,18,9]]_{5}$&$42$&$[[82,50,\geq\textbf{12}]]_{8}$&Theorem \ref{theorem-qcode2}&$[[82,50,10]]_{8}$\\

$3$&$[[44,8,\geq\textbf{13}]]_{5}$&Theorem \ref{theorem-qcode2}&$[[44,8,12]]_{5}$&$43$&$[[82,46,\geq\textbf{13}]]_{8}$&Theorem \ref{theorem-qcode2}&$[[82,46,11]]_{8}$\\

$4$&$[[44,6,\geq\textbf{14}]]_{5}$&Theorem \ref{theorem-qcode2}&$[[44,6,13]]_{5}$&$44$&$[[82,44,\geq\textbf{14}]]_{8}$&Theorem \ref{theorem-qcode2}&$[[82,44,12]]_{8}$\\

$5$&$[[50,20,\geq\textbf{11}]]_{5}$&Theorem \ref{theorem-qcode2}&$[[50,20,10]]_{5}$&$45$&$[[82,40,\geq\textbf{15}]]_{8}$&Theorem \ref{theorem-qcode2}&$[[82,40,13]]_{8}$\\

$6$&$[[50,14,\geq\textbf{13}]]_{5}$&Theorem \ref{theorem-qcode2}&$[[50,14,12]]_{5}$&$46$&$[[82,38,\geq\textbf{16}]]_{8}$&Theorem \ref{theorem-qcode2}&$[[82,38,14]]_{8}$\\

$7$&$[[50,8,\geq\textbf{15}]]_{5}$&Theorem \ref{theorem-qcode2}&$[[50,8,14]]_{5}$&$47$&$[[82,34,\geq\textbf{17}]]_{8}$&Theorem \ref{theorem-qcode2}&$[[82,34,16]]_{8}$\\

$8$&$[[50,2,\geq\textbf{17}]]_{5}$&Theorem \ref{theorem-qcode2}&$[[50,2,16]]_{5}$&$48$&$[[82,32,\geq\textbf{18}]]_{8}$&Theorem \ref{theorem-qcode2}&$[[82,32,16]]_{8}$\\

$9$&$[[66,34,\geq\textbf{12}]]_{7}$&Theorem \ref{theorem-qcode2}&$[[66,34,11]]_{7}$&$49$&$[[82,28,\geq\textbf{19}]]_{8}$&Theorem \ref{theorem-qcode2}&$[[82,28,18]]_{8}$\\

$10$&$[[66,28,\geq\textbf{14}]]_{7}$&Theorem \ref{theorem-qcode2}&$[[66,28,13]]_{7}$&$50$&$[[82,26,\geq\textbf{20}]]_{8}$&Theorem \ref{theorem-qcode2}&$[[82,26,19]]_{8}$\\

$11$&$[[66,22,\geq\textbf{16}]]_{7}$&Theorem \ref{theorem-qcode2}&$[[66,22,15]]_{7}$&$51$&$[[82,22,\geq\textbf{21}]]_{8}$&Theorem \ref{theorem-qcode2}&$[[82,22,20]]_{8}$\\

$12$&$[[66,16,\geq\textbf{18}]]_{7}$&Theorem \ref{theorem-qcode2}&$[[66,16,17]]_{7}$&$52$&$[[82,20,\geq\textbf{22}]]_{8}$&Theorem \ref{theorem-qcode2}&$[[82,20,21]]_{8}$\\

$13$&$[[68,36,\geq\textbf{12}]]_{7}$&Theorem \ref{theorem-qcode2}&$[[68,36,11]]_{7}$&$53$&$[[87,63,\geq\textbf{8}]]_{8}$&Theorem \ref{theorem-qcode3}&$[[87,63,7]]_{8}$\\

$14$&$[[68,30,\geq\textbf{14}]]_{7}$&Theorem \ref{theorem-qcode2}&$[[68,30,13]]_{7}$&$54$&$[[88,38,\geq\textbf{18}]]_{8}$&Theorem \ref{theorem-qcode2}&$[[88,38,16]]_{8}$\\

$15$&$[[68,24,\geq\textbf{16}]]_{7}$&Theorem \ref{theorem-qcode2}&$[[68,24,15]]_{7}$&$55$&$[[88,34,\geq\textbf{19}]]_{8}$&Theorem \ref{theorem-qcode2}&$[[88,34,18]]_{8}$\\

$16$&$[[70,26,\geq\textbf{16}]]_{7}$&Theorem \ref{theorem-qcode2}&$[[70,26,15]]_{7}$&$56$&$[[88,32,\geq\textbf{20}]]_{8}$&Theorem \ref{theorem-qcode2}&$[[88,32,18]]_{8}$\\

$17$&$[[70,20,\geq\textbf{18}]]_{7}$&Theorem \ref{theorem-qcode2}&$[[70,20,17]]_{7}$&$57$&$[[88,28,\geq\textbf{21}]]_{8}$&Theorem \ref{theorem-qcode2}&$[[88,28,20]]_{8}$\\

$18$&$[[72,22,\geq\textbf{18}]]_{7}$&Theorem \ref{theorem-qcode2}&$[[72,22,17]]_{7}$&$58$&$[[88,26,\geq\textbf{22}]]_{8}$&Theorem \ref{theorem-qcode2}&$[[88,26,21]]_{8}$\\

$19$&$[[72,16,\geq\textbf{20}]]_{7}$&Theorem \ref{theorem-qcode2}&$[[72,16,19]]_{7}$&$59$&$[[88,22,\geq\textbf{23}]]_{8}$&Theorem \ref{theorem-qcode2}&$[[88,22,22]]_{8}$\\

$20$&$[[74,42,\geq\textbf{12}]]_{7}$&Theorem \ref{theorem-qcode2}&$[[74,42,11]]_{7}$&$60$&$[[88,20,\geq\textbf{24}]]_{8}$&Theorem \ref{theorem-qcode2}&$[[88,20,23]]_{8}$\\

$21$&$[[74,36,\geq\textbf{14}]]_{7}$&Theorem \ref{theorem-qcode2}&$[[74,36,13]]_{7}$&$61$&$[[93,65,\geq\textbf{9}]]_{8}$&Theorem \ref{theorem-qcode3}&$[[93,65,8]]_{8}$\\

$22$&$[[74,30,\geq\textbf{16}]]_{7}$&Theorem \ref{theorem-qcode2}&$[[74,30,15]]_{7}$&$62$&$[[94,44,\geq\textbf{18}]]_{8}$&Theorem \ref{theorem-qcode2}&$[[94,44,16]]_{8}$\\

$23$&$[[74,24,\geq\textbf{18}]]_{7}$&Theorem \ref{theorem-qcode2}&$[[74,24,17]]_{7}$&$63$&$[[94,38,\geq\textbf{20}]]_{8}$&Theorem \ref{theorem-qcode2}&$[[94,38,18]]_{8}$\\

$24$&$[[74,18,\geq\textbf{20}]]_{7}$&Theorem \ref{theorem-qcode2}&$[[74,18,19]]_{7}$&$64$&$[[94,32,\geq\textbf{22}]]_{8}$&Theorem \ref{theorem-qcode2}&$[[94,32,20]]_{8}$\\

$25$&$[[76,44,\geq\textbf{12}]]_{7}$&Theorem \ref{theorem-qcode2}&$[[76,44,10]]_{7}$&$65$&$[[94,26,\geq\textbf{24}]]_{8}$&Theorem \ref{theorem-qcode2}&$[[94,26,23]]_{8}$\\

$26$&$[[76,38,\geq\textbf{14}]]_{7}$&Theorem \ref{theorem-qcode2}&$[[76,38,13]]_{7}$&$66$&$[[96,52,\geq\textbf{16}]]_{8}$&Theorem \ref{theorem-qcode2}&$[[96,52,14]]_{8}$\\

$27$&$[[76,32,\geq\textbf{16}]]_{7}$&Theorem \ref{theorem-qcode2}&$[[76,32,15]]_{7}$&$67$&$[[96,46,\geq\textbf{18}]]_{8}$&Theorem \ref{theorem-qcode2}&$[[96,46,16]]_{8}$\\

$28$&$[[76,26,\geq\textbf{18}]]_{7}$&Theorem \ref{theorem-qcode2}&$[[76,26,17]]_{7}$&$68$&$[[96,40,\geq\textbf{20}]]_{8}$&Theorem \ref{theorem-qcode2}&$[[96,40,18]]_{8}$\\

$29$&$[[76,20,\geq\textbf{20}]]_{7}$&Theorem \ref{theorem-qcode2}&$[[76,20,19]]_{7}$&$69$&$[[96,34,\geq\textbf{22}]]_{8}$&Theorem \ref{theorem-qcode2}&$[[96,34,20]]_{8}$\\

$30$&$[[78,46,\geq\textbf{12}]]_{7}$&Theorem \ref{theorem-qcode2}&$[[78,46,10]]_{7}$&$70$&$[[96,28,\geq\textbf{24}]]_{8}$&Theorem \ref{theorem-qcode2}&$[[96,28,23]]_{8}$\\

$31$&$[[78,40,\geq\textbf{14}]]_{7}$&Theorem \ref{theorem-qcode2}&$[[78,40,12]]_{7}$&$71$&$[[96,22,\geq\textbf{26}]]_{8}$&Theorem \ref{theorem-qcode2}&$[[96,22,25]]_{8}$\\

$32$&$[[78,34,\geq\textbf{16}]]_{7}$&Theorem \ref{theorem-qcode2}&$[[78,34,15]]_{7}$&$72$&$[[98,60,\geq\textbf{14}]]_{8}$&Theorem \ref{theorem-qcode2}&$[[98,60,12]]_{8}$\\

$33$&$[[78,28,\geq\textbf{18}]]_{7}$&Theorem \ref{theorem-qcode2}&$[[78,28,17]]_{7}$&$73$&$[[98,54,\geq\textbf{16}]]_{8}$&Theorem \ref{theorem-qcode2}&$[[98,54,14]]_{8}$\\

$34$&$[[78,22,\geq\textbf{20}]]_{7}$&Theorem \ref{theorem-qcode2}&$[[78,22,19]]_{7}$&$74$&$[[98,48,\geq\textbf{18}]]_{8}$&Theorem \ref{theorem-qcode2}&$[[98,48,16]]_{8}$\\

$35$&$[[80,48,\geq\textbf{12}]]_{7}$&Theorem \ref{theorem-qcode2}&$[[80,48,10]]_{7}$&$75$&$[[98,42,\geq\textbf{20}]]_{8}$&Theorem \ref{theorem-qcode2}&$[[98,42,18]]_{8}$\\

$36$&$[[80,42,\geq\textbf{14}]]_{7}$&Theorem \ref{theorem-qcode2}&$[[80,42,12]]_{7}$&$76$&$[[98,36,\geq\textbf{22}]]_{8}$&Theorem \ref{theorem-qcode2}&$[[98,36,20]]_{8}$\\

$37$&$[[80,36,\geq\textbf{16}]]_{7}$&Theorem \ref{theorem-qcode2}&$[[80,36,15]]_{7}$&$77$&$[[98,30,\geq\textbf{24}]]_{8}$&Theorem \ref{theorem-qcode2}&$[[98,30,23]]_{8}$\\

$38$&$[[80,30,\geq\textbf{18}]]_{7}$&Theorem \ref{theorem-qcode2}&$[[80,30,17]]_{7}$&$78$&$[[98,24,\geq\textbf{26}]]_{8}$&Theorem \ref{theorem-qcode2}&$[[98,24,25]]_{8}$\\

$39$&$[[80,24,\geq\textbf{20}]]_{7}$&Theorem \ref{theorem-qcode2}&$[[80,24,19]]_{7}$&$79$&$[[99,75,\geq\textbf{8}]]_{8}$&Theorem \ref{theorem-qcode3}&$[[99,75,7]]_{8}$\\

$40$&$[[80,18,\geq\textbf{22}]]_{7}$&Theorem \ref{theorem-qcode2}&$[[80,18,21]]_{7}$&$80$&$[[99,71,\geq\textbf{9}]]_{8}$&Theorem \ref{theorem-qcode3}&$[[99,71,8]]_{8}$\\

\hline
\end{tabular}
\end{table}

\begin{table}[!htbp]
\renewcommand\arraystretch{1.3}
\centering	
\scriptsize
\setlength{\abovecaptionskip}{0.cm}
\setlength{\belowcaptionskip}{0.2cm}
\caption{Comparison of our $q$-ary quantum codes using the lengthening/subcode construction with those in Grassl's database \cite{Grassl2025Bounds}}\label{table2}
\vspace{3pt}
\begin{tabular}{ccc|ccc}
\hline
No.&Lengthening/Subcode&Record in \cite{Grassl2025Bounds}&No.&Lengthening/Subcode&Record in \cite{Grassl2025Bounds}\\
\hline
$1/2$&$[[31,6,\geq\textbf{10}]]_{4}/[[30,5,\geq\textbf{10}]]_{4}$&$[[31,6,9]]_{4}/[[30,5,9]]_{4}$&$81/82$&
$[[82,57,\geq8]]_{8}/[[81,56,\geq8]]_{8}$&$[[82,57,8]]_{8}/[[81,56,8]]_{8}$\\

$3/4$&$[[45,18,\geq\textbf{10}]]_{5}/[[44,17,\geq\textbf{10}]]_{5}$&$[[45,18,9]]_{5}/[[44,17,9]]_{5}$&$83/84$&
$[[83,50,\geq\textbf{12}]]_{8}/[[82,49,\geq\textbf{12}]]_{8}$&$[[83,50,10]]_{8}/[[82,49,10]]_{8}$\\

$5/6$&$[[45,8,\geq\textbf{13}]]_{5}/[[44,7,\geq13]]_{5}$&$[[45,8,12]]_{5}/[[44,7,13]]_{5}$&$85/86$&
$[[83,46,\geq\textbf{13}]]_{8}/[[82,45,\geq\textbf{13}]]_{8}$&$[[83,46,12]]_{8}/[[82,45,12]]_{8}$\\

$7/8$&$[[45,6,\geq\textbf{14}]]_{5}/[[44,5,\geq14]]_{5}$&$[[45,6,13]]_{5}/[[44,5,14]]_{5}$&$87/88$&
$[[83,44,\geq\textbf{14}]]_{8}/[[82,43,\geq\textbf{14}]]_{8}$&$[[83,44,12]]_{8}/[[82,43,12]]_{8}$\\

$9/10$&$[[51,20,\geq\textbf{11}]]_{5}/[[50,19,\geq\textbf{11}]]_{5}$&$[[51,20,10]]_{5}/[[50,19,10]]_{5}$&$89/90$&
$[[83,40,\geq\textbf{15}]]_{8}/[[82,39,\geq\textbf{15}]]_{8}$&$[[83,40,14]]_{8}/[[82,39,14]]_{8}$\\

$11/12$&$[[51,14,\geq\textbf{13}]]_{5}/[[50,13,\geq\textbf{13}]]_{5}$&$[[51,14,12]]_{5}/[[50,13,12]]_{5}$&$91/92$&
$[[83,38,\geq\textbf{16}]]_{8}/[[82,37,\geq\textbf{16}]]_{8}$&$[[83,38,15]]_{8}/[[82,37,15]]_{8}$\\

$13/14$&$[[51,8,\geq\textbf{15}]]_{5}/[[50,7,\geq\textbf{15}]]_{5}$&$[[51,8,14]]_{5}/[[50,7,14]]_{5}$&$93/94$&
$[[83,34,\geq\textbf{17}]]_{8}/[[82,33,\geq\textbf{17}]]_{8}$&$[[83,34,16]]_{8}/[[82,33,16]]_{8}$\\

$15/16$&$[[51,2,\geq17]]_{5}/[[50,1,\geq17]]_{5}$&$[[51,2,17]]_{5}/[[50,1,17]]_{5}$&$95/96$&
$[[83,32,\geq\textbf{18}]]_{8}/[[82,31,\geq\textbf{18}]]_{8}$&$[[83,32,17]]_{8}/[[82,31,17]]_{8}$\\

$17/18$&$[[67,34,\geq\textbf{12}]]_{7}/[[66,33,\geq\textbf{12}]]_{7}$&$[[67,34,11]]_{7}/[[66,33,11]]_{7}$&$97/98$&
$[[83,28,\geq\textbf{19}]]_{8}/[[82,27,\geq\textbf{19}]]_{8}$&$[[83,28,18]]_{8}/[[82,27,18]]_{8}$\\

$19/20$&$[[67,28,\geq\textbf{14}]]_{7}/[[66,27,\geq\textbf{14}]]_{7}$&$[[67,28,13]]_{7}/[[66,27,13]]_{7}$&$99/100$&
$[[83,26,\geq\textbf{20}]]_{8}/[[82,25,\geq\textbf{20}]]_{8}$&$[[83,26,19]]_{8}/[[82,25,19]]_{8}$\\

$21/22$&$[[67,22,\geq\textbf{16}]]_{7}/[[66,21,\geq\textbf{16}]]_{7}$&$[[67,22,15]]_{7}/[[66,21,15]]_{7}$&$101/102$&
$[[83,22,\geq21]]_{8}/[[82,21,\geq21]]_{8}$&$[[83,22,21]]_{8}/[[82,21,21]]_{8}$\\

$23/24$&$[[67,16,\geq\textbf{18}]]_{7}/[[66,15,\geq18]]_{7}$&$[[67,16,17]]_{7}/[[66,15,18]]_{7}$&$103/104$&
$[[83,20,\geq\textbf{22}]]_{8}/[[82,19,\geq\textbf{22}]]_{8}$&$[[83,20,21]]_{8}/[[82,19,21]]_{8}$\\

$25/26$&$[[69,36,\geq\textbf{12}]]_{7}/[[68,35,\geq\textbf{12}]]_{7}$&$[[69,36,11]]_{7}/[[68,35,11]]_{7}$&$105/106$&
$[[88,63,\geq8]]_{8}/[[87,62,\geq8]]_{8}$&$[[88,63,8]]_{8}/[[87,62,8]]_{8}$\\

$27/28$&$[[69,30,\geq\textbf{14}]]_{7}/[[68,29,\geq\textbf{14}]]_{7}$&$[[69,30,13]]_{7}/[[68,29,13]]_{7}$&$107/108$&
$[[89,38,\geq\textbf{18}]]_{8}/[[88,37,\geq\textbf{18}]]_{8}$&$[[89,38,17]]_{8}/[[88,37,17]]_{8}$\\

$29/30$&$[[69,24,\geq\textbf{16}]]_{7}/[[68,23,\geq\textbf{16}]]_{7}$&$[[69,24,15]]_{7}/[[68,23,15]]_{7}$&$109/110$&
$[[89,34,\geq\textbf{19}]]_{8}/[[88,33,\geq\textbf{19}]]_{8}$&$[[89,34,18]]_{8}/[[88,33,18]]_{8}$\\

$31/32$&$[[71,26,\geq\textbf{16}]]_{7}/[[70,25,\geq\textbf{16}]]_{7}$&$[[71,26,15]]_{7}/[[70,25,15]]_{7}$&$111/112$&
$[[89,32,\geq\textbf{20}]]_{8}/[[88,31,\geq\textbf{20}]]_{8}$&$[[89,32,19]]_{8}/[[88,31,19]]_{8}$\\

$33/34$&$[[71,20,\geq\textbf{18}]]_{7}/[[70,19,\geq\textbf{18}]]_{7}$&$[[71,20,17]]_{7}/[[70,19,17]]_{7}$&$113/114$&
$[[89,28,\geq\textbf{21}]]_{8}/[[88,27,\geq\textbf{21}]]_{8}$&$[[89,28,20]]_{8}/[[88,27,20]]_{8}$\\

$35/36$&$[[73,22,\geq\textbf{18}]]_{7}/[[72,21,\geq\textbf{18}]]_{7}$&$[[73,22,17]]_{7}/[[72,21,17]]_{7}$&$115/116$&
$[[89,26,\geq\textbf{22}]]_{8}/[[88,25,\geq\textbf{22}]]_{8}$&$[[89,26,21]]_{8}/[[88,25,21]]_{8}$\\

$37/38$&$[[73,16,\geq\textbf{20}]]_{7}/[[72,15,\geq20]]_{7}$&$[[73,16,19]]_{7}/[[72,15,20]]_{7}$&$117/118$&
$[[89,22,\geq23]]_{8}/[[88,21,\geq23]]_{8}$&$[[89,22,23]]_{8}/[[88,21,23]]_{8}$\\

$39/40$&$[[75,42,\geq\textbf{12}]]_{7}/[[74,41,\geq\textbf{12}]]_{7}$&$[[75,42,11]]_{7}/[[74,41,11]]_{7}$&$119/120$&
$[[89,20,\geq\textbf{24}]]_{8}/[[88,19,\geq\textbf{24}]]_{8}$&$[[89,20,23]]_{8}/[[88,19,23]]_{8}$\\

$41/42$&$[[75,36,\geq\textbf{14}]]_{7}/[[74,35,\geq\textbf{14}]]_{7}$&$[[75,36,13]]_{7}/[[74,35,13]]_{7}$&$121/122$&
$[[94,65,\geq9]]_{8}/[[93,64,\geq9]]_{8}$&$[[94,65,9]]_{8}/[[93,64,9]]_{8}$\\

$43/44$&$[[75,30,\geq\textbf{16}]]_{7}/[[74,29,\geq\textbf{16}]]_{7}$&$[[75,30,15]]_{7}/[[74,29,15]]_{7}$&$123/124$&
$[[95,44,\geq\textbf{18}]]_{8}/[[94,43,\geq\textbf{18}]]_{8}$&$[[95,44,16]]_{8}/[[94,43,16]]_{8}$\\

$45/46$&$[[75,24,\geq\textbf{18}]]_{7}/[[74,23,\geq\textbf{18}]]_{7}$&$[[75,24,17]]_{7}/[[74,23,17]]_{7}$&$125/126$&
$[[95,38,\geq\textbf{20}]]_{8}/[[94,37,\geq\textbf{20}]]_{8}$&$[[95,38,19]]_{8}/[[94,37,19]]_{8}$\\

$47/48$&$[[75,18,\geq\textbf{20}]]_{7}/[[74,17,\geq\textbf{20}]]_{7}$&$[[75,18,19]]_{7}/[[74,17,19]]_{7}$&$127/128$&
$[[95,32,\geq\textbf{22}]]_{8}/[[94,31,\geq\textbf{22}]]_{8}$&$[[95,32,21]]_{8}/[[94,31,21]]_{8}$\\

$49/50$&$[[77,44,\geq\textbf{12}]]_{7}/[[76,43,\geq\textbf{12}]]_{7}$&$[[77,44,10]]_{7}/[[76,43,11]]_{7}$&$129/130$&
$[[95,26,\geq\textbf{24}]]_{8}/[[94,25,\geq\textbf{24}]]_{8}$&$[[95,26,23]]_{8}/[[94,25,23]]_{8}$\\

$51/52$&$[[77,38,\geq\textbf{14}]]_{7}/[[76,37,\geq\textbf{14}]]_{7}$&$[[77,38,13]]_{7}/[[76,37,13]]_{7}$&$131/132$&
$[[97,52,\geq\textbf{16}]]_{8}/[[96,51,\geq\textbf{16}]]_{8}$&$[[97,52,14]]_{8}/[[96,51,14]]_{8}$\\

$53/54$&$[[77,32,\geq\textbf{16}]]_{7}/[[76,31,\geq\textbf{16}]]_{7}$&$[[77,32,15]]_{7}/[[76,31,15]]_{7}$&$133/134$&
$[[97,46,\geq\textbf{18}]]_{8}/[[96,45,\geq\textbf{18}]]_{8}$&$[[97,46,16]]_{8}/[[96,45,16]]_{8}$\\

$55/56$&$[[77,26,\geq\textbf{18}]]_{7}/[[76,25,\geq\textbf{18}]]_{7}$&$[[77,26,17]]_{7}/[[76,25,17]]_{7}$&$135/136$&
$[[97,40,\geq\textbf{20}]]_{8}/[[96,39,\geq\textbf{20}]]_{8}$&$[[97,40,18]]_{8}/[[96,39,18]]_{8}$\\

$57/58$&$[[77,20,\geq\textbf{20}]]_{7}/[[76,19,\geq\textbf{20}]]_{7}$&$[[77,20,19]]_{7}/[[76,19,19]]_{7}$&$137/138$&
$[[97,34,\geq\textbf{22}]]_{8}/[[96,33,\geq\textbf{22}]]_{8}$&$[[97,34,21]]_{8}/[[96,33,21]]_{8}$\\

$59/60$&$[[79,46,\geq\textbf{12}]]_{7}/[[78,45,\geq\textbf{12}]]_{7}$&$[[79,46,10]]_{7}/[[78,45,10]]_{7}$&$139/140$&
$[[97,28,\geq\textbf{24}]]_{8}/[[96,27,\geq\textbf{24}]]_{8}$&$[[97,28,23]]_{8}/[[96,27,23]]_{8}$\\

$61/62$&$[[79,40,\geq\textbf{14}]]_{7}/[[78,39,\geq\textbf{14}]]_{7}$&$[[79,40,12]]_{7}/[[78,39,13]]_{7}$&$141/142$&
$[[97,22,\geq\textbf{26}]]_{8}/[[96,21,\geq\textbf{26}]]_{8}$&$[[97,22,25]]_{8}/[[96,21,25]]_{8}$\\

$63/64$&$[[79,34,\geq\textbf{16}]]_{7}/[[78,33,\geq\textbf{16}]]_{7}$&$[[79,34,15]]_{7}/[[78,33,15]]_{7}$&$143/144$&
$[[99,60,\geq\textbf{14}]]_{8}/[[98,59,\geq\textbf{14}]]_{8}$&$[[99,60,12]]_{8}/[[98,59,12]]_{8}$\\

$65/66$&$[[79,28,\geq\textbf{18}]]_{7}/[[78,27,\geq\textbf{18}]]_{7}$&$[[79,28,17]]_{7}/[[78,27,17]]_{7}$&$145/146$&
$[[99,54,\geq\textbf{16}]]_{8}/[[98,53,\geq\textbf{16}]]_{8}$&$[[99,54,14]]_{8}/[[98,53,14]]_{8}$\\

$67/68$&$[[79,22,\geq\textbf{20}]]_{7}/[[78,21,\geq\textbf{20}]]_{7}$&$[[79,22,19]]_{7}/[[78,21,19]]_{7}$&$147/148$&
$[[99,48,\geq\textbf{18}]]_{8}/[[98,47,\geq\textbf{18}]]_{8}$&$[[99,48,16]]_{8}/[[98,47,16]]_{8}$\\

$69/70$&$[[81,48,\geq\textbf{12}]]_{7}/[[80,47,\geq\textbf{12}]]_{7}$&$[[81,48,11]]_{7}/[[80,47,10]]_{7}$&$149/150$&
$[[99,42,\geq\textbf{20}]]_{8}/[[98,41,\geq\textbf{20}]]_{8}$&$[[99,42,18]]_{8}/[[98,41,18]]_{8}$\\

$71/72$&$[[81,42,\geq\textbf{14}]]_{7}/[[80,41,\geq\textbf{14}]]_{7}$&$[[81,42,12]]_{7}/[[80,41,12]]_{7}$&$151/152$&
$[[99,36,\geq\textbf{22}]]_{8}/[[98,35,\geq\textbf{22}]]_{8}$&$[[99,36,21]]_{8}/[[98,35,21]]_{8}$\\

$73/74$&$[[81,36,\geq\textbf{16}]]_{7}/[[80,35,\geq\textbf{16}]]_{7}$&$[[81,36,15]]_{7}/[[80,35,15]]_{7}$&$153/154$&
$[[99,30,\geq\textbf{24}]]_{8}/[[98,29,\geq\textbf{24}]]_{8}$&$[[99,30,23]]_{8}/[[98,29,23]]_{8}$\\

$75/76$&$[[81,30,\geq\textbf{18}]]_{7}/[[80,29,\geq\textbf{18}]]_{7}$&$[[81,30,17]]_{7}/[[80,29,17]]_{7}$&$155/156$&
$[[99,24,\geq\textbf{26}]]_{8}/[[98,23,\geq\textbf{26}]]_{8}$&$[[99,24,25]]_{8}/[[98,23,25]]_{8}$\\

$77/78$&$[[81,24,\geq\textbf{20}]]_{7}/[[80,23,\geq\textbf{20}]]_{7}$&$[[81,24,19]]_{7}/[[80,23,19]]_{7}$&$157/158$&
$[[100,75,\geq\textbf{8}]]_{8}/[[99,74,\geq\textbf{8}]]_{8}$&$[[100,75,7]]_{8}/[[99,74,7]]_{8}$\\

$79/80$&$[[81,18,\geq\textbf{22}]]_{7}/[[80,17,\geq\textbf{22}]]_{7}$&$[[81,18,21]]_{7}/[[80,17,21]]_{7}$&$159/160$&
$[[100,71,\geq9]]_{8}/[[99,70,\geq9]]_{8}$&$[[100,71,9]]_{8}/[[99,70,9]]_{8}$\\
\hline
\end{tabular}
\vspace{4pt}
\begin{tablenotes}
\item $\bullet$ The lengthening construction reveals that an $[[n,k,d]]_q$ quantum code can yield an $[[n+1,k,\geq d]]_q$ quantum code if $k>0$,
as illustrated in Lemma \ref{proposition66} (1). The odd-numbered quantum codes in Table \ref{table2} are obtained by applying the lengthening construction to the quantum codes in Table \ref{table1}.

\vspace{2pt}

\item $\bullet$ The subcode construction reveals that an $[[n,k,d]]_q$ quantum code can yield an $[[n,k-1,\geq d]]_q$ quantum code if $k>1$,
as illustrated in Lemma \ref{proposition66} (2). The even-numbered quantum codes in Table \ref{table2} are obtained by applying the subcode construction to the quantum codes in Table \ref{table1}.

\vspace{2pt}

\item $\bullet$ These $142$ quantum codes (whose minimum distance lower bounds are marked in bold in Table \ref{table2}) outperform the best-known records maintained
in Grassl's database \cite{Grassl2025Bounds}.
\end{tablenotes}
\end{table}

\section{Optimal Pure Quantum $(r,\delta)$-LRCs}\label{qopt}

In this section, we propose two effective schemes for constructing optimal pure quantum $(r,\delta)$-LRCs from MP codes.
Based on these schemes, we explicitly construct four new infinite families of optimal pure quantum $(r,\delta)$-LRCs with flexible parameters.
In addition, we report an interesting phenomenon regarding optimal pure quantum $(r,\delta)$-LRCs derived from our framework.
To the best of our knowledge, the new discovery that quantum codes are simultaneously optimal pure quantum $(r,\delta)$-LRCs and record-breaking quantum codes has not been previously reported in the literature.

\subsection{Two Schemes for Constructing Optimal Pure Quantum $(r,\delta)$-LRCs From MP Codes}\label{subs-qopt}

In the following theorem, we propose our first scheme for constructing optimal pure quantum $(r,\delta)$-LRCs from MP codes
whose defining matrix is a $2\times 2$ $\tau$-OD matrix.

\begin{theorem}\label{theorem-optqx}
Let $\mathcal{C}(A)=[\mathcal{C}_1,\mathcal{C}_2]\cdot A$, where $A$ is a $2\times 2$ $\tau$-OD matrix over $\mathbb{F}_{q^{2}}$ with $\tau=(1,2)$, and $\mathcal{C}_i$ is an $[n,k_i,d(\mathcal{C}_i)]_{q^{2}}$ linear code for $i=1,2$. Suppose the following conditions are satisfied:
\begin{itemize}
\item [(i)] $\mathcal{C}_2\subseteq \mathcal{C}_1$ and $\mathcal{C}_2^{\perp_\mathrm{H}}\subseteq \mathcal{C}_1$;

\item [(ii)] $\mathcal{C}_1$ is an optimal classical $(r,\delta)$-LRC with $d(\mathcal{C}_1)=\delta$, and $\mathcal{C}_2$ is an MDS code with $d(\mathcal{C}_2)\leq 2\delta$;

\item [(iii)] $\left\lceil \frac{k_1+k_2}{r}\right\rceil=\left\lceil\frac{k_1}{r}\right\rceil+1$.
\end{itemize}
Then, the following statements hold.
\begin{itemize}
\item [(1)] $\mathcal{C}(A)$ is both a Hermitian dual-containing code and an optimal classical $(r,\delta)$-LRC.

\item [(2)] If $k_1+k_2=n$, or the following inequality holds:
\begin{align}\label{contain}
\mathrm{min}\{2d(\mathcal{C}_2^{\perp_\mathrm{H}}),d(\mathcal{C}_1^{\perp_\mathrm{H}})\}>d(\mathcal{C}_2),
\end{align}
then there exists an optimal pure quantum $(r,\delta)$-LRC with parameters $[[2n,2(k_1+k_2-n),d(\mathcal{C}_2)]]_{q}$.
\end{itemize}
\end{theorem}

\begin{IEEEproof}
It follows from Lemma \ref{proposition4} that
\begin{align*}
\mathcal{C}(A)^{\bot_{\mathrm{H}}}\subseteq\mathcal{C}(A).
\end{align*}
Hence, $\mathcal{C}(A)$ is a Hermitian dual-containing code.
By Lemma \ref{local}, we derive $d(\mathcal{C}(A))=d(\mathcal{C}_2)$. Consequently, $\mathcal{C}(A)$ has parameters $[2n,k_1+k_2,d(\mathcal{C}_2)]_{q^{2}}$.
As $\mathcal{C}_2$ is an MDS code, we have
\begin{align}\label{dd}
d(\mathcal{C}_2)=n-k_2+1.
\end{align}
Moreover, we have
\begin{align*}
d(\mathcal{C}_2)=2n-(k_1+k_2)+1-\left(\left\lceil\frac{k_1+k_2}{r}\right\rceil-1\right)(\delta-1),
\end{align*}
which implies that $\mathcal{C}(A)$ is an optimal classical $(r,\delta)$-LRC. This completes the proof of the statement (1).

Let us now prove that the statement (2) holds.
If $k_1+k_2=n$, then $\mathcal{C}(A)^{\perp_{\mathrm{H}}}=\mathcal{C}(A)$. Thus, $\mathcal{C}(A)$ induces a pure quantum code $\mathcal{Q}_{1}$ with parameters
$[[2n,0,d(\mathcal{C}_2)]]_{q}$. Suppose that inequality \eqref{contain} holds. It follows from Lemma \ref{nested-distance} that
\begin{align*}
d(\mathcal{C}(A)^{\bot_{\mathrm{H}}})>d(\mathcal{C}(A)).
\end{align*}
As a result, $\mathcal{C}(A)$ induces a pure quantum code $\mathcal{Q}_{2}$ with parameters $[[2n,2(k_1+k_2-n),d(\mathcal{C}_2)]]_{q}$.
Since $\mathcal{C}(A)$ is an optimal classical $(r,\delta)$-LRC, it follows that $\mathcal{Q}_{1}$ and $\mathcal{Q}_{2}$ are optimal pure quantum $(r,\delta)$-LRCs. This completes the proof of the statement (2).

Therefore, we complete the proof.
\end{IEEEproof}

\begin{corollary}\label{cor-optq}
Let $\mathcal{C}(A)=[\mathcal{C}_1,\mathcal{C}_2]\cdot A$, where $A$ is a $2\times 2$ $\tau$-OD matrix over $\mathbb{F}_{q^{2}}$ with $\tau=(1,2)$, $\mathcal{C}_1$ is an $[r+\delta-1,r,\delta]_{q^{2}}$ MDS code, and $\mathcal{C}_2$ is an $[r+\delta-1,k_{2},r+\delta-k_2]_{q^{2}}$ MDS code. Suppose the following conditions are satisfied:
\begin{itemize}
\item [(i)] $\mathcal{C}_2\subseteq \mathcal{C}_1$ and $\mathcal{C}_2^{\perp_\mathrm{H}}\subseteq \mathcal{C}_1$;

\item [(ii)] $r\leq\delta+k_2$.
\end{itemize}
Then, the following statements hold.
\begin{itemize}
\item [(1)] $\mathcal{C}(A)$ is both a Hermitian dual-containing code and an optimal classical $(r,\delta)$-LRC.

\item [(2)] If one of the following conditions holds:
\begin{itemize}
\item $k_2=\delta-1$;

\item $k_2>\delta-1$ and $3k_2>r+\delta-2$,
\end{itemize}
then there exists an optimal pure quantum $(r,\delta)$-LRC with parameters $[[2(r+\delta-1),2(k_2-\delta+1),r+\delta-k_2]]_{q}$.
\end{itemize}
\end{corollary}

\vspace{6pt}

The following theorem presents our second scheme for constructing optimal pure quantum $(r,\delta)$-LRCs from MP codes whose defining matrix is a $3\times 3$
$\tau$-OD matrix.

\begin{theorem}\label{theorem-optqx-x1}
Let $\mathcal{C}(A)=[\mathcal{C}_1,\mathcal{C}_1,\mathcal{C}_2]\cdot A$, where $A$ is a $3\times 3$ $\tau$-OD matrix over $\mathbb{F}_{q^{2}}$ with $\tau=(2,3)$,
and $\mathcal{C}_i$ is an $[n,k_i,d(\mathcal{C}_i)]_{q^{2}}$ linear code for $i=1,2$. Suppose the following conditions are satisfied:
\begin{itemize}
\item [(i)] $\mathcal{C}_2\subseteq \mathcal{C}_1$ and $\mathcal{C}_2^{\perp_\mathrm{H}}\subseteq \mathcal{C}_1$;

\item [(ii)] $\mathcal{C}_1$ is an optimal classical $(r,\delta)$-LRC with $d(\mathcal{C}_1)=\delta$, and $\mathcal{C}_2$ is an MDS code with
$d(\mathcal{C}_2)\leq 2\delta$;

\item [(iii)] $\left\lceil \frac{2k_1+k_2}{r}\right\rceil=2\left\lceil\frac{k_1}{r}\right\rceil+1$.
\end{itemize}
Then, the following statements hold.
\begin{itemize}
\item [(1)] $\mathcal{C}(A)$ is both a Hermitian dual-containing code and an optimal classical $(r,\delta)$-LRC.

\item [(2)] If $4k_1+2k_2=3n$, or the following inequality holds:
\begin{align}\label{2inequ-pure}
\mathrm{min}\{2d(\mathcal{C}_2^{\perp_\mathrm{H}}),d(\mathcal{C}_1^{\perp_\mathrm{H}})\}>d(\mathcal{C}_2),
\end{align}
then there exists an optimal pure quantum $(r,\delta)$-LRC with parameters $[[3n,4k_1+2k_2-3n,d(\mathcal{C}_2)]]_{q}$.
\end{itemize}
\end{theorem}

\begin{IEEEproof}
It follows from Lemma \ref{proposition4} that
\begin{align*}
\mathcal{C}(A)^{\bot_{\mathrm{H}}}\subseteq\mathcal{C}(A).
\end{align*}
Hence, $\mathcal{C}(A)$ is a Hermitian dual-containing code. Moreover, we derive $d(\mathcal{C}(A))=d(\mathcal{C}_2)$. Consequently, $\mathcal{C}(A)$ has parameters $[3n,2k_1+k_2,d(\mathcal{C}_2)]_{q^{2}}$. As $\mathcal{C}_2$ is an MDS code, we have
\begin{align}\label{2dd}
d(\mathcal{C}_2)=n-k_2+1.
\end{align}
Moreover, we derive
\begin{align*}
d(\mathcal{C}_2)=3n-(2k_1+k_2)+1-\left(\left\lceil\frac{2k_1+k_2}{r}\right\rceil-1\right)(\delta-1),
\end{align*}
which implies that $\mathcal{C}(A)$ is an optimal classical $(r,\delta)$-LRC. This completes the proof of the statement (1).

Let us now prove that the statement (2) holds.
If $4k_1+2k_2=3n$, then $\mathcal{C}(A)^{\perp_{\mathrm{H}}}=\mathcal{C}(A)$. Thus, $\mathcal{C}(A)$ induces a pure quantum code $\mathcal{Q}_{1}$ with parameters
$[[3n,0,d(\mathcal{C}_2)]]_{q}$. Suppose that inequality \eqref{2inequ-pure} holds. It follows that
\begin{align*}
d(\mathcal{C}(A)^{\bot_{\mathrm{H}}})>d(\mathcal{C}(A)).
\end{align*}
As a result, $\mathcal{C}(A)$ induces a pure quantum code $\mathcal{Q}_{2}$ with parameters $[[3n,4k_1+2k_2-3n,d(\mathcal{C}_2)]]_{q}$.
Since $\mathcal{C}(A)$ is an optimal classical $(r,\delta)$-LRC, it follows that $\mathcal{Q}_{1}$ and $\mathcal{Q}_{2}$ are optimal pure quantum $(r,\delta)$-LRCs. This completes the proof of the statement (2).

Therefore, we complete the proof.
\end{IEEEproof}

\begin{corollary}\label{cor-optq-x1}
Let $\mathcal{C}(A)=[\mathcal{C}_1,\mathcal{C}_1,\mathcal{C}_2]\cdot A$, where $A$ is a $3\times 3$ $\tau$-OD matrix over $\mathbb{F}_{q^{2}}$ with $\tau=(2,3)$, $\mathcal{C}_1$ is an $[r+\delta-1,r,\delta]_{q^{2}}$ MDS code, and $\mathcal{C}_2$ is an $[r+\delta-1,k_{2},r+\delta-k_2]_{q^{2}}$ MDS code. Suppose the following conditions are satisfied:
\begin{itemize}
\item [(i)] $\mathcal{C}_2\subseteq \mathcal{C}_1$ and $\mathcal{C}_2^{\perp_\mathrm{H}}\subseteq \mathcal{C}_1$;

\item [(ii)] $r\leq\delta+k_2$.
\end{itemize}
Then, the following statements hold.
\begin{itemize}
\item [(1)] $\mathcal{C}(A)$ is both a Hermitian dual-containing code and an optimal classical $(r,\delta)$-LRC.

\item [(2)] If one of the following conditions holds:
\begin{itemize}
\item $2k_2=3\delta-r-3$;

\item $k_2>\delta-1$ and $3k_2>r+\delta-2$,
\end{itemize}
then there exists an optimal pure quantum $(r,\delta)$-LRC with parameters $[[3(r+\delta-1),2k_2-3\delta+r+3,r+\delta-k_2]]_{q}$.
\end{itemize}
\end{corollary}

\subsection{New Infinite Families of Optimal Pure Quantum $(r,\delta)$-LRCs}\label{infinitequantum}

In this subsection, we apply Theorems \ref{theorem-optqx} and \ref{theorem-optqx-x1} to construct four new infinite families of optimal pure quantum $(r,\delta)$-LRCs with flexible parameters.

\begin{lemma}\label{lem-contain}
Let $(n-1)\mid(q^2-1)$ and $\mathrm{char}(\mathbb{F}_{q^2})\mid n$. Let $\alpha\in\mathbb{F}_{q^2}$ be a primitive $(n-1)$-th root of unity.
Denote $\mathbf{a}=(0,1,\alpha,\ldots,\alpha^{n-2})$ and $\mathbf{1}=(1,1,\ldots,1)$ as two row vectors of length $n$.
Then, $\mathrm{GRS}_{k_2}(\mathbf{a},\mathbf{1})^{\bot_{\mathrm{H}}}\subseteq \mathrm{GRS}_{k_1}(\mathbf{a},\mathbf{1})$ if and only if the following relation holds in $\mathbb{Z}_{n-1}$:
\begin{align*}
\{q,2q,\ldots,(n-k_2-1)q\}\subseteq \{1,2,\ldots,k_1-1\}.
\end{align*}
\end{lemma}

\begin{IEEEproof}
Let $\mathbf{u}_{i}$ and $\mathbf{v}_j$ be the $i$-th row of $G_{n-k}(\mathbf{a}^{q},\mathbf{1})$ and the $j$-th row of $G_{k}(\mathbf{a}^{q},\mathbf{1})$, respectively, where $1\leq i\leq n-k$ and $1\leq j\leq k$. It follows that
\begin{align*}
\mathrm{GRS}_{n-k}(\mathbf{a}^{q},\mathbf{1})\subseteq \mathrm{GRS}_{k}(\mathbf{a}^{q},\mathbf{1})^{\bot_{\mathrm{E}}}
=\mathrm{GRS}_{k}(\mathbf{a},\mathbf{1})^{\bot_{\mathrm{H}}}.
\end{align*}

Then,
$\mathrm{GRS}_{n-k}(\mathbf{a}^{q},\mathbf{1})=\mathrm{GRS}_{k}(\mathbf{a},\mathbf{1})^{\bot_{\mathrm{H}}}$. Hence,
$\mathrm{GRS}_{k_2}(\mathbf{a},\mathbf{1})^{\bot_{\mathrm{H}}}\subseteq \mathrm{GRS}_{k_1}(\mathbf{a},\mathbf{1})$ if and only if
$\mathrm{GRS}_{n-k_2}(\mathbf{a}^{q},\mathbf{1})\subseteq \mathrm{GRS}_{k_1}(\mathbf{a},\mathbf{1})$. Equivalently, the following relation holds in $\mathbb{Z}_{n-1}$:
\begin{align*}
\{q,2q,\ldots,(n-k_2-1)q\}\subseteq \{1,2,\ldots,k_1-1\},
\end{align*}
which completes the proof.
\end{IEEEproof}

\begin{remark}\label{remark1x}
It is worth noting that when $\mathrm{char}(\mathbb{F}_{q^2})=2$, the condition $(n-1)\mid(q^2-1)$ implies that $\mathrm{char}(\mathbb{F}_{q^2})\mid n$ is automatically satisfied. Consequently, in this case, the conditions $(n-1)\mid(q^2-1)$ and $\mathrm{char}(\mathbb{F}_{q^2})\mid n$ in Lemma \ref{lem-contain} can be simplified to
$(n-1)\mid(q^2-1)$.
\end{remark}

\begin{lemma}\label{cor-sepecial}
With the same conditions and notation as in Lemma \ref{lem-contain}, the following two statements hold.
\begin{itemize}
\item [(1)] If $n=q$, then $\mathrm{GRS}_{k_2}(\mathbf{a},\mathbf{1})^{\bot_{\mathrm{H}}}\subseteq \mathrm{GRS}_{k_1}(\mathbf{a},\mathbf{1})$ if and only if
$k_1+k_2\geq q$.

\item [(2)] If $n=q^2$, $0\leq l_1\leq q-2$ and $0 \leq l_2\leq 2q-2$ , then
$\mathrm{GRS}_{q^2-l_2-1}(\mathbf{a},\mathbf{1})^{\bot_{\mathrm{H}}}\subseteq \mathrm{GRS}_{q^2-l_1-1}(\mathbf{a},\mathbf{1})$.
\end{itemize}
\end{lemma}

\begin{IEEEproof}
It follows from Lemma \ref{lem-contain}.
\end{IEEEproof}

\begin{remark}
Hermitian self-orthogonal GRS codes (or, Hermitian dual-containing GRS codes) are very useful for constructing quantum MDS codes
(see, e.g., \cite{Fang2019Some,Grassl2004On,Jin2010Application,Jin2014A,Li2008Quantum}). Lemma \ref{cor-sepecial} characterizes the conditions under which the GRS codes $\mathcal{C}_1$ and $\mathcal{C}_2$, with dimensions determined by the parameters $k_1$ (resp. $l_1$) and $k_2$ (resp. $l_2$), constitute a Hermitian dual-containing code pair $(\mathcal{C}_1,\mathcal{C}_2)$, namely, $\mathcal{C}_2^{\perp_\mathrm{H}}\subseteq \mathcal{C}_1$.
\end{remark}

Let us present our first infinite family of optimal pure quantum $(r,\delta)$-LRCs in the following theorem.

\begin{theorem}\label{cor-sepoql}
Let $1\leq k_2\leq  k_1\leq q-1$, $k_1+k_2\geq q$ and $2k_{1}-k_{2}\leq q+1$, where $q$ is a prime power. Then, there exists an infinite family of optimal pure quantum $(k_1,q-k_1+1)$-LRCs with parameters
\begin{align*}
[[2q,2(k_1+k_2-q),q-k_2+1]]_{q}.
\end{align*}
\end{theorem}

\begin{IEEEproof}
In Corollary \ref{cor-optq}, we set $\mathcal{C}_1=\mathrm{GRS}_{k_1}(\mathbf{a},\mathbf{1})$, $\mathcal{C}_2=\mathrm{GRS}_{k_2}(\mathbf{a},\mathbf{1})$, $n=q$, $r=k_1$ and $\delta=q-k_1+1$, where $\mathbf{a}=(0,1,\alpha,\ldots,\alpha^{n-2})$ and $\mathbf{1}=(1,1,\ldots,1)$ are row vectors of length $n$ with $\alpha\in\mathbb{F}_{q^2}$ being a primitive $(n-1)$-th root of unity.

The conditions (i) and (ii) in Corollary \ref{cor-optq} are equivalent to the following conditions
\begin{align}\label{cond0}
k_2\leq  k_1, \ k_1+k_2\geq q \ \mathrm{and}\ 2k_{1}-k_{2}\leq q+1.
\end{align}
Besides, the two conditions in the statement (2) of Corollary \ref{cor-optq} are respectively equivalent to the conditions
\begin{align}\label{cond1}
k_1+k_2=q,
\end{align}
\begin{align}\label{cond2}
k_1+k_2>q\ \mathrm{and} \ 3k_2>q-1.
\end{align}

Therefore, under the conditions \eqref{cond0}, it follows from Corollary \ref{cor-optq} that the following MP code
\begin{align*}
\mathcal{C}(A):=[\mathrm{GRS}_{k_1}(\mathbf{a},\mathbf{1}),\mathrm{GRS}_{k_2}(\mathbf{a},\mathbf{1})]\cdot A
\end{align*}
is both a Hermitian dual-containing code and an optimal classical $(k_1,q-k_1+1)$-LRC. Moreover, it induces an optimal pure quantum $(k_1,q-k_1+1)$-LRC with parameters $[[2q,2(k_1+k_2-q),q-k_2+1]]_{q}$.

Therefore, we complete the proof.
\end{IEEEproof}

\vspace{6pt}

Next, applying Lemma \ref{cor-sepecial} (2) and Corollary \ref{cor-optq}, we show our second infinite family of optimal pure quantum $(r,\delta)$-LRCs in the following theorem.

\begin{theorem}\label{qsquare}
Let $0\leq l_1\leq q-2$ and $l_1 \leq l_2\leq 2l_1+2$, where $q>2$ is a prime power. Then, there exists an infinite family of optimal pure quantum $(q^2-l_1-1,l_1+2)$-LRCs with parameters
\begin{align*}
[[2q^2,2(q^2-l_1-l_2-2),l_2+2]]_{q}.
\end{align*}
\end{theorem}

\begin{IEEEproof}
In Corollary \ref{cor-optq}, we set $\mathcal{C}_1=\mathrm{GRS}_{q^2-l_1-1}(\mathbf{a},\mathbf{1})$, $\mathcal{C}_2=\mathrm{GRS}_{q^2-l_2-1}(\mathbf{a},\mathbf{1})$, $n=q^2$, $k_2=q^2-l_2-1$, $r=q^2-l_1-1$ and $\delta=l_1+2$, where $\mathbf{a}=(0,1,\alpha,\ldots,\alpha^{n-2})$ and $\mathbf{1}=(1,1,\ldots,1)$ are row vectors of length $n$ with $\alpha\in\mathbb{F}_{q^2}$ being a primitive $(n-1)$-th root of unity.

We have $\mathcal{C}_2\subseteq \mathcal{C}_1$, $l_2\leq 2q-2$, and thus $\mathcal{C}_2^{\bot_{\mathrm{H}}}\subseteq \mathcal{C}_1$.
Thus, the condition (i) of Corollary \ref{cor-optq} is satisfied.

Note also that the conditions $0\leq l_1\leq q-2$, $l_2\leq 2l_1+2$ and $q>2$ imply that
\begin{align}\label{20261}
l_1+l_2<q^2-2,
\end{align}
\begin{align}\label{20262}
3l_2<2q^{2}-2.
\end{align}

Therefore, the following MP code
\begin{align*}
\mathcal{C}(A):=[\mathrm{GRS}_{q^2-l_1-1}(\mathbf{a},\mathbf{1}),\mathrm{GRS}_{q^2-l_2-1}(\mathbf{a},\mathbf{1})]\cdot A
\end{align*}
is both a Hermitian dual-containing code and an optimal classical $(q^2-l_1-1,l_1+2)$-LRC. Furthermore, it induces an optimal pure quantum $(q^2-l_1-1,l_1+2)$-LRC with parameters $[[2q^2,2(q^2-l_1-l_2-2),l_2+2]]_{q}$.

Therefore, we complete the proof.
\end{IEEEproof}

\vspace{6pt}

In the following theorem, we construct our third infinite family of optimal pure quantum $(r,\delta)$-LRCs.

\begin{theorem}\label{qsquaree}
Let $1\leq k_2\leq  k_1\leq q-1$ and $2k_{1}-k_{2}\leq q+1$, where $q>2$ is a prime power. If one of the following conditions holds:
\begin{itemize}
\item $k_1=k_2=\frac{q}{2}$, where $q$ is even;

\item $k_1+k_2>q$,
\end{itemize}
then there exists an infinite family of optimal pure quantum $(k_1,q-k_1+1)$-LRCs with parameters
\begin{align*}
[[3q,4k_1+2k_2-3q,q-k_2+1]]_{q}.
\end{align*}
\end{theorem}

\begin{IEEEproof}
In Corollary \ref{cor-optq-x1}, we set $\mathcal{C}_1=\mathrm{GRS}_{k_1}(\mathbf{a},\mathbf{1})$, $\mathcal{C}_2=\mathrm{GRS}_{k_2}(\mathbf{a},\mathbf{1})$, $n=q$, $r=k_1$ and $\delta=q-k_1+1$, where $\mathbf{a}=(0,1,\alpha,\ldots,\alpha^{n-2})$ and $\mathbf{1}=(1,1,\ldots,1)$ are row vectors of length $n$ with $\alpha\in\mathbb{F}_{q^2}$ being a primitive $(n-1)$-th root of unity.

The conditions (i) and (ii) in Corollary \ref{cor-optq-x1} are equivalent to the following conditions
\begin{align}\label{cond5}
k_2\leq  k_1, \ k_1+k_2\geq q \ \mathrm{and}\ 2k_{1}-k_{2}\leq q+1.
\end{align}
Furthermore, the two conditions in the statement (2) of Corollary \ref{cor-optq-x1} are respectively equivalent to the conditions
\begin{align}\label{cond6}
4k_1+2k_2=3q,
\end{align}
\begin{align*}
k_1+k_2>q\ \mathrm{and} \ 3k_2>q-1.
\end{align*}

It is verified that the conditions \eqref{cond5} and \eqref{cond6} are equivalent to the following condition
\begin{align}\label{cond8}
k_1=k_2=\frac{q}{2},\ \text{where} \ q \ \text{is even}.
\end{align}
Therefore, the following MP code
\begin{align*}
\mathcal{C}(A):=[\mathrm{GRS}_{k_1}(\mathbf{a},\mathbf{1}),\mathrm{GRS}_{k_1}(\mathbf{a},\mathbf{1}),\mathrm{GRS}_{k_2}(\mathbf{a},\mathbf{1})]\cdot A
\end{align*}
is both a Hermitian dual-containing code and an optimal classical $(k_1,q-k_1+1)$-LRC, and it induces an optimal pure quantum $(k_1,q-k_1+1)$-LRC with parameters $[[3q,4k_1+2k_2-3q,q-k_2+1]]_{q}$.

This completes the proof.
\end{IEEEproof}

\vspace{6pt}

Finally, we present our fourth infinite family of optimal pure quantum $(r,\delta)$-LRCs in the following theorem.

\begin{theorem}\label{3qsquare}
Let $0\leq l_1\leq q-2$ and $l_1 \leq l_2\leq 2l_1+2$, where $q>2$ is a prime power.
Then, there exists an infinite family of optimal pure quantum $(q^2-l_1-1,l_1+2)$-LRCs with parameters
\begin{align*}
[[3q^2,3q^2-4l_1-2l_2-6,l_2+2]]_{q}.
\end{align*}
\end{theorem}

\begin{IEEEproof}
In Corollary \ref{cor-optq-x1}, we set $\mathcal{C}_1=\mathrm{GRS}_{q^2-l_1-1}(\mathbf{a},\mathbf{1})$, $\mathcal{C}_2=\mathrm{GRS}_{q^2-l_2-1}(\mathbf{a},\mathbf{1})$, $n=q^2$, $k_2=q^2-l_2-1$, $r=q^2-l_1-1$ and $\delta=l_1+2$ , where $\mathbf{a}=(0,1,\alpha,\ldots,\alpha^{n-2})$ and $\mathbf{1}=(1,1,\ldots,1)$ are row vectors of length $n$ with $\alpha\in\mathbb{F}_{q^2}$ being a primitive $(n-1)$-th root of unity.

Similar to the proof of Theorem \ref{qsquare}, it can be verified that the conditions $l_1\leq q-2$, $l_1 \leq l_2\leq 2l_1+2$ and $q>2$ ensure that conditions (i), (ii), and the second case of condition (2) in Corollary \ref{cor-optq-x1} are satisfied. Therefore, the following MP code
\begin{align*}
\mathcal{C}(A):=[\mathrm{GRS}_{q^2-l_1-1}(\mathbf{a},\mathbf{1}),\mathrm{GRS}_{q^2-l_1-1}(\mathbf{a},\mathbf{1}),\mathrm{GRS}_{q^2-l_2-1}(\mathbf{a},\mathbf{1})]\cdot A
\end{align*}
is both a Hermitian dual-containing code and an optimal classical $(q^2-l_1-1,l_1+2)$-LRC. Furthermore, it induces an optimal pure quantum $(q^2-l_1-1,l_1+2)$-LRC with parameters $[[3q^2,3q^2-4l_1-2l_2-6,l_2+2]]_{q}$.

This completes the proof.
\end{IEEEproof}

\subsection{Comparison of Our Optimal Pure Quantum $(r,\delta)$-LRCs with Existing Ones}\label{compar}

\begin{table}[!htbp]
\renewcommand\arraystretch{1.6}
\centering	
\footnotesize
\setlength{\abovecaptionskip}{0.cm}
\setlength{\belowcaptionskip}{0.3cm}
\caption{A Summary of the Infinite Families of $[[n,k,d]]_{q}$ Optimal Pure Quantum $(r,\delta)$-LRCs from Existing Works and This Paper}\label{table666}
\vspace{3pt}
\begin{tabular}{c>{\centering\arraybackslash}p{6.0cm}c>{\centering\arraybackslash}p{5.2cm}c}
\toprule
No.&$[[n,k,d]]_{q}$&$(r,\delta)$&Conditions&Reference\\
\midrule
$1$&$[[n_{1}n_{2},2(i+1)(j-1)-n_{1}n_{2},(n_{1}-i)(n_{2}-j)]]_{q}$&$(i+1,n_{1}-i)$&\makecell[l]{$q=p^{a}$, $(n_i-1)\mid(q-1)$ for $i=1,2$,\\
$p\mid \mathrm{gcd}(n_1,n_2)$, $i>\frac{n_1}{2}$, $j=n_2-1$}&\cite[Prop. 35]{Galindo2026}\\
\midrule

$2$&$[[n_{1}n_{2},2(i+1)(j-1)-n_{1}n_{2},(n_{1}-i)(n_{2}-j)]]_{q}$&$(j+1,n_{2}-j)$&\makecell[l]{$q=p^{a}$, $(n_i-1)\mid(q-1)$ for $i=1,2$,\\
$p\mid \mathrm{gcd}(n_1,n_2)$, $i=n_1-1$, $j>\frac{n_2}{2}$}&\cite[Prop. 35]{Galindo2026}\\
\midrule

$3$&$[[n_{1}n_{2},2[(i+1)(n_2-1)+s+1]-n_{1}n_{2},n_{1}-s]]_{q}$&$(i+1,n_{1}-i)$&\makecell[l]{$q=p^{a}$, $(n_i-1)\mid(q-1)$ for $i=1,2$,\\
$p\mid \mathrm{gcd}(n_1,n_2)$, $\frac{n_1-1}{2}<i\leq n_1-2$, \\
$i>s\geq \mathrm{max}\{n_1-i-1,2i-n_1\}$}&\cite[Prop. 36]{Galindo2026}\\
\midrule

$4$&$[[n_{1}n_{2},2[(j+1)(n_1-1)+s+1]-n_{1}n_{2},n_{2}-s]]_{q}$&$(j+1,n_{2}-j)$&\makecell[l]{$q=p^{a}$, $(n_i-1)\mid(q-1)$ for $i=1,2$,\\
$p\mid \mathrm{gcd}(n_1,n_2)$, $\frac{n_2-1}{2}<j\leq n_2-2$, \\
$j>s\geq \mathrm{max}\{n_2-j-1,2j-n_2\}$}&\cite[Prop. 36]{Galindo2026}\\
\midrule

$5$&$[[h^{2},2[(h-i)(h-1)+(h-j)]-h^{2},j+1]]_{q}$&$(h-i,i+1)$&\makecell[l]{$(h-1)\mid(q-1)$, $1-h$ is a square in $\mathbb{F}_{q}$,\\
$q$ is odd, $\frac{j-1}{2}\leq i\leq j<\frac{h}{2}$,\\
$h^{2}-2i(h-1)-2j\geq 0$}&\cite[Cor. 38]{Galindo2025Optimal}\\
\midrule

$6$&$[[q^{2},2[q(q-t)-(d-t-1)]-q^{2},d]]_{q}$&$(q-t,t+1)$&\makecell[l]{$t+1\leq d\leq \mathrm{min}\big\{2(t+1),\frac{q^{2}-t(2q-2)+2}{2}\big\}$, \\
$q$ is odd, $1\leq t<\frac{q}{2}$}&\cite[Thm. 39]{Galindo2025Optimal}\\
\midrule

$7$&$[[mh,2[m(h-t)-(d-t-1)]-mh,d]]_{q}$&$(h-t,t+1)$&\makecell[l]{$m,h\leq q$, $(h-1)\mid(q-1)$, $q$ is odd,\\
$1-h$ is a square in $\mathbb{F}_{q}$, $1\leq t<\frac{h}{2}$, \\
$t+1\leq d\leq \mathrm{min}\big\{2(t+1),\frac{mh-t(2m-2)+2}{2}\big\}$}&\cite[Thm. 40]{Galindo2025Optimal}\\
\midrule

$8$&$[[q^{4},2[(q^{2}-a)(q^{2}-1)+(q^{2}-b)]-q^{4},b+1]]_{q}$&$(q^{2}-a,a+1)$&$0\leq a,b<q-1$, $a\leq b\leq 2a$, $q$ is odd
&\cite[Cor. 44]{Galindo2025Optimal}\\
\midrule

$9$&$[[\lambda(q^{2}+1),\lambda(q^{2}+1)-4\lambda u,2u+1]]_{q}$&$(q^2+1-2u,2u+1)$&\makecell[l]{$1\leq u\leq \frac{q}{2}$, $q$ is even,\\
$\lambda(q^{2}+1)\mid (q^{s}-1)$ for some $s=2\varsigma$}&\cite[Thm. 29]{Galindo2026QuantumBCH}\\
\midrule

$10$&$[[\lambda(q-1),\lambda(q-1)-2\lambda v,v+1]]_{q}$&$(q-1-v,v+1)$&\makecell[l]{$1\leq v\leq \frac{q-2}{2}$, $q\geq 3$,\\
$\lambda(q-1)\mid (q^{s}-1)$ for some $s=2\varsigma$}&\cite[Thm. 31]{Galindo2026QuantumBCH}\\
\midrule

$11$&$[[\lambda(q^{2}-1),\lambda(q^{2}-1)-2\lambda v,v+1]]_{q}$&$(q^2-1-v,v+1)$&\makecell[l]{$1\leq v\leq \frac{q^2-2}{q+1}$,\\
$\lambda(q^{2}-1)\mid (q^{s}-1)$ for some $s=2\varsigma$}&\cite[Thm. 31]{Galindo2026QuantumBCH}\\
\midrule

$12$&$[[\lambda(q^{2}+1)n_{2}\cdots n_{w},\lambda(q^{2}+1)n_{2}\cdots n_{w}-4\lambda un_{2}\cdots n_{w},2u+1]]_{q}$&$(q^2+1-2u,2u+1)$
&\makecell[l]{$\lambda(q^{2}+1)\mid (q^{s}-1)$ for some $s=2\varsigma$, \\
$2\leq n_{\ell}<q^2$ for $2\leq \ell\leq w$,\\
$q$ is even, $u\leq \frac{q}{2}$}&\cite[Thm. 34]{Galindo2026QuantumBCH}\\
\midrule

$13$&$[[\lambda(q-1)n_{2}\cdots n_{w},\lambda(q-1)n_{2}\cdots n_{w}-2\lambda vn_{2}\cdots n_{w},v+1]]_{q}$&$(q-1-v,v+1)$
&\makecell[l]{$2\leq n_{\ell}<q$ for $2\leq \ell\leq w$,\\
$\lambda(q-1)\mid (q^{s}-1)$, $v\leq \frac{q-2}{2}$}&\cite[Thm. 34]{Galindo2026QuantumBCH}\\
\midrule

$14$&$[[\lambda(q^{2}-1)n_{2}\cdots n_{w},\lambda(q^{2}-1)n_{2}\cdots n_{w}-2\lambda vn_{2}\cdots n_{w},v+1]]_{q}$&$(q^2-1-v,v+1)$
&\makecell[l]{$\lambda(q^{2}-1)\mid (q^{s}-1)$, $s=2\varsigma>2$, \\
$2\leq n_{\ell}<q^2$ for $2\leq \ell\leq w$, $v\leq \frac{q^2-2}{q+1}$,\\
or, $1\leq v\leq 2q-3$, $\lambda=2$}&\cite[Thm. 34]{Galindo2026QuantumBCH}\\
\midrule

$15$&$[[2q,2(k_1+k_2-q),q-k_2+1]]_{q}$&$(k_1,q-k_1+1)$&\makecell[l]{$1\leq k_2\leq  k_1\leq q-1$, $k_1+k_2\geq q$,\\
$2k_{1}-k_{2}\leq q+1$}&Thm. \ref{cor-sepoql}\\
\midrule

$16$&$[[2q^2,2(q^2-l_1-l_2-2),l_2+2]]_{q}$&$(q^2-l_1-1,l_1+2)$&$q>2$, $0\leq l_1\leq q-2$, $l_1 \leq l_2\leq 2l_1+2$&Thm. \ref{qsquare}\\
\midrule

$17$&$[[3q,4k_1+2k_2-3q,q-k_2+1]]_{q}$&$(k_1,q-k_1+1)$&\makecell[l]{$1\leq k_2\leq  k_1\leq q-1$, $2k_{1}-k_{2}\leq q+1$,\\
$k_1=k_2=\frac{q}{2}$ for even $q$, or $k_1+k_2>q$}&Thm. \ref{qsquaree}\\
\midrule

$18$&$[[3q^2,3q^2-4l_1-2l_2-6,l_2+2]]_{q}$&$(q^2-l_1-1,l_1+2)$&$q>2$, $0\leq l_1\leq q-2$, $l_1 \leq l_2\leq 2l_1+2$&Thm. \ref{3qsquare}\\
\bottomrule
\end{tabular}
\end{table}

Recently, Galindo et al. initiated and further developed the theory of quantum $(r,\delta)$-LRCs in a series of works
\cite{Galindo2026, Galindo2025Optimal, Galindo2026QuantumBCH}. In particular, they constructed infinite families of optimal pure quantum $(r,\delta)$-LRCs using different methods. Table \ref{table666} summarizes the infinite families of optimal pure quantum $(r,\delta)$-LRCs from \cite{Galindo2026, Galindo2025Optimal, Galindo2026QuantumBCH} and this paper. A comparison between these studies, in terms of both construction methods and code parameters, is provided below:
\begin{itemize}
\item In \cite{Galindo2026}, Galindo, Hernando, Mart\'{i}n-Cruz, and Matsumoto constructed several infinite families of optimal pure quantum $(r,\delta)$-LRCs via
$\emptyset$-affine variety codes, as listed in Nos. 1-4 of Table \ref{table666}. In Theorems \ref{cor-sepoql}-\ref{3qsquare}, we employ the matrix product construction technique to derive optimal pure quantum $(r,\delta)$-LRCs, where the constituent codes of MP codes are GRS codes. Note that the code lengths in Nos. 1-4 are of the form $n_{1}n_{2}$ with $(n_i-1)\mid(q-1)$ for $i=1,2$. Therefore, the optimal pure quantum $(r,\delta)$-LRCs presented in Theorems \ref{cor-sepoql}-\ref{3qsquare} are different from those listed in Nos. 1-4 of Table \ref{table666}.

\vspace{3pt}

\item In \cite{Galindo2025Optimal}, Galindo, Hernando, Munuera, and Ruano presented several infinite families of optimal pure quantum $(r,\delta)$-LRCs using MP
codes, as listed in Nos. 5-8 of Table \ref{table666}. Note that the MP codes employed in this work and those in \cite{Galindo2025Optimal} utilize defining matrices of different sizes. As a result, the parameters of our optimal pure quantum $(r,\delta)$-LRCs constructed in Theorems \ref{cor-sepoql}-\ref{3qsquare} are different from those in Nos. 5-8 of Table \ref{table666}.

\vspace{3pt}

\item In \cite{Galindo2026QuantumBCH}, Galindo, Hernando, and Matsumoto provided several infinite families of optimal pure quantum $(r,\delta)$-LRCs from
BCH and homothetic-BCH codes, as listed in Nos. 9-14 of Table \ref{table666}. Note that the code lengths listed in Nos. 9-14 are constant multiples of $q^2 \pm 1$ or $q-1$. In contrast, the lengths of our codes in Theorems \ref{cor-sepoql}-\ref{3qsquare} are constant multiples of $q^2$ or $q$. Consequently, the optimal pure quantum $(r,\delta)$-LRCs established in Theorems \ref{cor-sepoql}-\ref{3qsquare} constitute new code families.
\end{itemize}

\subsection{Optimal Pure Quantum $(r,\delta)$-LRCs Can Also Be Best-Known, Optimal, or Record-Breaking Quantum Codes}

In Theorems \ref{cor-sepoql}-\ref{3qsquare}, we constructed four new infinite families of optimal pure quantum $(r,\delta)$-LRCs.
In this subsection, we report an interesting phenomenon: our scheme can yield quantum codes that are simultaneously optimal pure quantum $(r,\delta)$-LRCs and best-known, optimal, or record-breaking quantum codes.

We begin with the following examples.

\begin{example}\label{ex-qopt4}
Let $q=4$, $l_1=2$ and $l_2=6$ in Theorem \ref{3qsquare}. Then, we obtain an optimal pure quantum $(13,4)$-LRC with parameters
\begin{align*}
[[48,22,8]]_{4}.
\end{align*}
According to Grassl's database \cite{Grassl2025Bounds}, this code is also a best-known quantum code.
\end{example}

\begin{example}\label{ex-qopt2}
Let $q=7$, $l_1=5$ and $l_2=12$ in Theorem \ref{qsquare}. Then, we obtain an optimal pure quantum $(43,7)$-LRC with parameters
\begin{align*}
[[98,60,\textbf{14}]]_{7}.
\end{align*}
According to Grassl's database \cite{Grassl2025Bounds}, this code is also a record-breaking quantum code, which improves the minimum distance of the best-known $[[98,60,13]]_{7}$ quantum code in \cite{Grassl2025Bounds}.
\end{example}

\begin{table}[!htbp]
\renewcommand\arraystretch{1.3}
\centering	
\footnotesize
\setlength{\abovecaptionskip}{0.cm}
\setlength{\belowcaptionskip}{0.1cm}
\caption{Comparison of our optimal pure quantum $(r,\delta)$-LRCs with best-known quantum codes in Grassl's database \cite{Grassl2025Bounds}}\label{table3}
\vspace{3pt}
\begin{tabular}{cccccccccccc}
\hline
No.&$q$&$l_{1}$&$l_{2}$&$k_{1}$&$k_{2}$&Our quantum code&$(r,\delta)$&Reference&Record in \cite{Grassl2025Bounds}\\

\hline

$1$&$3$&$0$&$1$&$-$&$-$&$[[18,12,3]]_{3}^{\circ}$&(8,2)&Theorem \ref{qsquare}&$[[18,12,3]]_{3}$\\
$2$&$3$&$0$&$2$&$-$&$-$&$[[18,10,4]]_{3}^{\circ}$&(8,2)&Theorem \ref{qsquare}&$[[18,10,4]]_{3}$\\
$3$&$3$&$1$&$3$&$-$&$-$&$[[18,6,5]]_{3}^{\ast}$&(7,3)&Theorem \ref{qsquare}&$[[18,6,5]]_{3}$\\
$4$&$3$&$1$&$4$&$-$&$-$&$[[18,4,6]]_{3}^{\ast}$&(7,3)&Theorem \ref{qsquare}&$[[18,4,6]]_{3}$\\
$5$&$3$&$0$&$2$&$-$&$-$&$[[27,17,4]]_{3}^{\ast}$&(8,2)&Theorem \ref{3qsquare}&$[[27,17,4]]_{3}$\\
\hline

$6$&$4$&$-$&$-$&$3$&$1$&$[[8,0,4]]_{4}^{\circ}$&(3,2)&Theorem \ref{cor-sepoql}&$[[8,0,4]]_{4}$\\
$7$&$4$&$0$&$2$&$-$&$-$&$[[32,24,4]]_{4}^{\circ}$&(15,2)&Theorem \ref{qsquare}&$[[32,24,4]]_{4}$\\
$8$&$4$&$1$&$3$&$-$&$-$&$[[32,20,5]]_{4}^{\ast}$&(14,3)&Theorem \ref{qsquare}&$[[32,20,5]]_{4}$\\
$9$&$4$&$1$&$4$&$-$&$-$&$[[32,18,6]]_{4}^{\ast}$&(14,3)&Theorem \ref{qsquare}&$[[32,18,6]]_{4}$\\
$10$&$4$&$2$&$5$&$-$&$-$&$[[32,14,7]]_{4}^{\ast}$&(13,4)&Theorem \ref{qsquare}&$[[32,14,7]]_{4}$\\
$11$&$4$&$2$&$6$&$-$&$-$&$[[32,12,8]]_{4}^{\ast}$&(13,4)&Theorem \ref{qsquare}&$[[32,12,8]]_{4}$\\
$12$&$4$&$0$&$2$&$-$&$-$&$[[48,38,4]]_{4}^{\ast}$&(15,2)&Theorem \ref{3qsquare}&$[[48,38,4]]_{4}$\\
$13$&$4$&$2$&$6$&$-$&$-$&$[[48,22,8]]_{4}^{\ast}$&(13,4)&Theorem \ref{3qsquare}&$[[48,22,8]]_{4}$\\
\hline

$14$&$5$&$1$&$3$&$-$&$-$&$[[50,38,5]]_{5}^{\ast}$&(23,3)&Theorem \ref{qsquare}&$[[50,38,5]]_{5}$\\
$15$&$5$&$1$&$4$&$-$&$-$&$[[50,36,6]]_{5}^{\ast}$&(23,3)&Theorem \ref{qsquare}&$[[50,36,6]]_{5}$\\
$16$&$5$&$2$&$5$&$-$&$-$&$[[50,32,7]]_{5}^{\ast}$&(22,4)&Theorem \ref{qsquare}&$[[50,32,7]]_{5}$\\
$17$&$5$&$2$&$6$&$-$&$-$&$[[50,30,8]]_{5}^{\ast}$&(22,4)&Theorem \ref{qsquare}&$[[50,30,8]]_{5}$\\
$18$&$5$&$3$&$7$&$-$&$-$&$[[50,26,9]]_{5}^{\ast}$&(21,5)&Theorem \ref{qsquare}&$[[50,26,9]]_{5}$\\
$19$&$5$&$3$&$8$&$-$&$-$&$[[50,24,\textbf{10}]]_{5}^{\diamond}$&(21,5)&Theorem \ref{qsquare}&$[[50,24,9]]_{5}$\\
$20$&$5$&$1$&$4$&$-$&$-$&$[[75,57,6]]_{5}^{\ast}$&(23,3)&Theorem \ref{3qsquare}&$[[75,57,6]]_{5}$\\
$21$&$5$&$2$&$6$&$-$&$-$&$[[75,49,8]]_{5}^{\ast}$&(22,4)&Theorem \ref{3qsquare}&$[[75,49,8]]_{5}$\\
$22$&$5$&$3$&$8$&$-$&$-$&$[[75,41,10]]_{5}^{\ast}$&(21,5)&Theorem \ref{3qsquare}&$[[75,41,10]]_{5}$\\
\hline

$23$&$7$&$2$&$5$&$-$&$-$&$[[98,80,7]]_{7}^{\ast}$&(46,4)&Theorem \ref{qsquare}&$[[98,80,7]]_{7}$\\
$24$&$7$&$2$&$6$&$-$&$-$&$[[98,78,8]]_{7}^{\ast}$&(46,4)&Theorem \ref{qsquare}&$[[98,78,8]]_{7}$\\
$25$&$7$&$3$&$7$&$-$&$-$&$[[98,74,9]]_{7}^{\ast}$&(45,5)&Theorem \ref{qsquare}&$[[98,74,9]]_{7}$\\
$26$&$7$&$3$&$8$&$-$&$-$&$[[98,72,10]]_{7}^{\ast}$&(45,5)&Theorem \ref{qsquare}&$[[98,72,10]]_{7}$\\
$27$&$7$&$4$&$9$&$-$&$-$&$[[98,68,11]]_{7}^{\ast}$&(44,6)&Theorem \ref{qsquare}&$[[98,68,11]]_{7}$\\
$28$&$7$&$4$&$10$&$-$&$-$&$[[98,66,12]]_{7}^{\ast}$&(44,6)&Theorem \ref{qsquare}&$[[98,66,12]]_{7}$\\
$29$&$7$&$5$&$11$&$-$&$-$&$[[98,62,\textbf{13}]]_{7}^{\diamond}$&(43,7)&Theorem \ref{qsquare}&$[[98,62,12]]_{7}$\\
$30$&$7$&$5$&$12$&$-$&$-$&$[[98,60,\textbf{14}]]_{7}^{\diamond}$&(43,7)&Theorem \ref{qsquare}&$[[98,60,13]]_{7}$\\
\hline
\end{tabular}
\end{table}

\begin{remark}\label{remark1rc}
We report an interesting phenomenon by exhibiting $30$ optimal pure quantum $(r,\delta)$-LRCs derived from our framework; that is, there exist quantum codes that are not only optimal pure quantum $(r,\delta)$-LRCs but also, according to Grassl's database \cite{Grassl2025Bounds}, best-known, optimal, or record-breaking quantum codes.
To the best of our knowledge, the new discovery that quantum codes are simultaneously optimal pure quantum $(r,\delta)$-LRCs and record-breaking quantum codes has not been previously reported in the literature. Specifically, as shown in Table \ref{table3},
\begin{itemize}
\item [1)] The optimal pure quantum $(r,\delta)$-LRCs, denoted by $[[n,k,d]]_{q}^{\ast}$ (Nos. $3$-$5$, $8$-$18$, and $20$-$28$), also match the best-known parameters
in Grassl's database \cite{Grassl2025Bounds}.

\item [2)] The optimal pure quantum $(r,\delta)$-LRCs, denoted by $[[n,k,d]]_{q}^{\circ}$ (Nos. $1$, $2$, $6$, and $7$), also attain the optimal parameters in Grassl's
database \cite{Grassl2025Bounds}.

\item [3)] The optimal pure quantum $(r,\delta)$-LRCs, denoted by $[[n,k,\bm{d}]]_{q}^{\diamond}$ (Nos. $19$, $29$, and $30$), are also the record-breaking quantum
codes according to Grassl's database \cite{Grassl2025Bounds}.
\end{itemize}
\end{remark}

\section{Concluding Remarks}\label{conclusion}

In this paper, we established a unified $\tau$-monomial decomposition theorem for invertible self-adjoint matrices over finite fields of arbitrary characteristic.
We proved the existence of $\tau$-OD matrices over $\mathbb{F}_{q^2}$ for any characteristic and demonstrated that there exist several new infinite families of $\tau$-OD matrices over $\mathbb{F}_{q^2}$ of characteristic $2$. As an application, we constructed several infinite families of quantum codes and presented $222$ record-breaking quantum codes that surpass the best-known records maintained in Grassl's database \cite{Grassl2025Bounds}.

Based on Theorems \ref{theorem-optqx} and \ref{theorem-optqx-x1}, we constructed four new infinite families of optimal pure quantum $(r,\delta)$-LRCs with flexible parameters in Theorems \ref{cor-sepoql}-\ref{3qsquare}. In addition, we reported an interesting phenomenon: there exist quantum codes that are not only optimal pure quantum $(r,\delta)$-LRCs but also, according to Grassl's database \cite{Grassl2025Bounds}, best-known, optimal, or record-breaking quantum codes.
To the best of our knowledge, the new discovery that quantum codes are simultaneously optimal pure quantum $(r,\delta)$-LRCs and record-breaking quantum codes has not been previously reported in the literature.

\end{document}